\documentclass[11pt]{article}

\usepackage[margin=1in]{geometry}
\usepackage{amsmath,amsthm,amssymb}
\usepackage{tikz}
\usepackage{microtype}
\usepackage[hidelinks]{hyperref}
\providecommand{\Description}[1]{}

\title{Auction Design with ROI-Constrained Bidders:\\
Truthfulness and Revenue Maximization}
\author{%
  \normalsize Zhiqiang Zhuang\textsuperscript{1}, Quan Yu\textsuperscript{1},
  Yisong Wang\textsuperscript{2}, Kewen Wang\textsuperscript{3}, Zhe Wang\textsuperscript{3}\\[0.6em]
  \small \textsuperscript{1}Qiannan Normal University for Nationalities\\
  \small \textsuperscript{2}Guizhou University\\
  \small \textsuperscript{3}Griffith University\\[0.4em]
  \small Zhiqiang Zhuang: \href{mailto:zoldlady@gmail.com}{\texttt{zoldlady@gmail.com}}%
}
\date{}

\DeclareMathOperator*{\argmax}{argmax}
\numberwithin{equation}{section}
\newtheorem{definition}{Definition}
\newtheorem{theorem}{Theorem}
\newtheorem{lemma}{Lemma}

\begin{document}
\raggedbottom
\maketitle

\begin{abstract}
The return-on-investment (ROI) constraint is central to many auctions,
particularly in online advertising, where a bidder is unwilling to pay
more than a fixed fraction of the value obtained. We study truthful and
revenue-maximizing auctions for ROI-constrained bidders. We first
characterize truthful auctions when both valuations and ROI constraints
are private, showing that the allocation rule uniquely determines the
payment rule. Building on this characterization, we introduce
\(\sigma\)-increment mechanisms that approach the best deterministic
truthful revenue and obtain at least a \(1/\bar r\) fraction of the
optimal revenue among all truthful mechanisms. For a single bidder, we
show that every truthful mechanism is weakly dominated by a convex
pricing function. We derive optimal pricing when either the valuation or the ROI constraint
is public.
\end{abstract}

\clearpage
\tableofcontents
\clearpage

\section{Introduction}

Online advertising auctions are a major revenue source for platforms such as Google, where advertisers compete repeatedly for limited display opportunities. Unlike bidders in classical auction environments, bidders in these markets are typically subject to \emph{return-on-investment} (ROI) constraints: their payments cannot exceed a fixed fraction of the value they obtain from the auction \cite{conf/Auerbach2008,conf/www/WilkensCN17,conf/www/GolrezaeiLL21}. ROI constraints are also relevant beyond online advertising, including in auctions for treasury securities, emissions permits, and energy \cite{pmlr-v267-golrezaei25a}. Such constraints change bidding behavior and substantially complicate auction design.

Efficiency, revenue maximization, and incentive compatibility are well understood in classical quasilinear settings, but ROI constraints make each of these objectives considerably harder to achieve. Since a bidder is now characterized by both a valuation and an ROI constraint, the resulting type space is multidimensional, and this multidimensionality undermines many of the standard tools of mechanism design: monotonicity properties are no longer immediate, incentive constraints become more intricate, and the familiar correspondence between allocation and payment rules no longer extends directly.

We begin by characterizing truthful mechanisms when both the valuation and the ROI constraint are private. The analysis rests on two concepts: the \emph{unit payment}, the payment per unit of allocation, and the \emph{unit-payment cap} \(c=v/r\), the highest unit payment that a bidder with valuation \(v\) and ROI constraint \(r\) can afford. Working in the transformed type space \((v,c)\)---replacing the ROI constraint with the unit-payment cap---lets us separate affordability from utility comparisons, since a report is affordable exactly when its unit payment does not exceed the bidder's cap. This separation is what lets us derive the clean allocation properties that truthfulness requires. Examining one-dimensional slices of the transformed type space then yields a closed-form payment identity, showing that, despite the two-dimensional type space, the allocation rule alone determines the payment rule.

Building on this characterization, we turn to deterministic truthful mechanisms, whose allocations and payments turn out to depend essentially only on the unit-payment cap. Under the standard regularity assumption on the distributions of unit-payment cap, Myerson's optimal single-item auction---applied to unit-payment caps rather than valuations---gives a tight upper bound on the expected revenue. We introduce $\sigma$-increment mechanisms, which perturb the Myerson winning thresholds and payments by a small increment. These mechanisms are truthful and approach the revenue bound as \(\sigma\) tends to zero; moreover, their revenue achieves at least a \(1/\bar r\) fraction of the optimal expected revenue among all truthful mechanisms, where \(\bar r\) is the largest possible ROI constraint.

Finally, we consider the single-bidder setting and show that every truthful mechanism can be replaced by a continuous, convex, and nondecreasing pricing function over allocation probabilities without lowering any type's payment. We derive the optimal pricing functions when either the valuation or the ROI constraint is public: the public-valuation price is zero up to a cutoff and then rises linearly, while the public-ROI price follows a power law under DMR.

\section{Auction Model}
\label{sec:auction-model}

A seller auctions a single item to $n$ \emph{risk-neutral},
\emph{utility-maximizing} bidders. Bidder $i$ has a valuation
$v_i \in [0,\bar v]$ and an ROI constraint $r_i \in [1,\bar r]$, where
$\bar r>1$.
The 2-tuple $(v_i,r_i)$ constitutes bidder $i$'s \emph{type} $t_i$.
We write $T$ for the set of all types.

An \emph{auction mechanism} $(x,p)$ consists of an \emph{allocation
function} $x:T^n \to [0,1]^n$ and a \emph{payment function}
$p:T^n \to \mathbb{R}^n$. Given reported types
$\mathbf{t}=(t_1,\ldots,t_n)$, $x_i(\mathbf{t})$ is the probability that
bidder $i$ receives the item, with $\sum_{i=1}^n x_i(\mathbf{t}) \leq 1$,
and $p_i(\mathbf{t})$ is bidder $i$'s payment.

For a bidder $i$ with type $t_i=(v_i,r_i)$, her ROI under mechanism
$(x,p)$ is $v_i x_i(\mathbf t)/p_i(\mathbf t)$ when $p_i(\mathbf t)>0$, and
$\infty$ otherwise; the ROI constraint is violated whenever her ROI is
below $r_i$. Bidder $i$'s utility is accordingly defined as
\[
u(\mathbf{t};t_i) =
\begin{cases}
v_i x_i(\mathbf{t}) - p_i(\mathbf{t}) & \text{if } p_i(\mathbf{t}) \leq (v_i / r_i)\, x_i(\mathbf{t}), \\
-\infty & \text{otherwise}.
\end{cases}
\]
That is, $u(\mathbf{t};t_i)$ takes the usual quasilinear form whenever the
ROI constraint is satisfied, and drops to negative infinity whenever it
is violated.

For $x_i(\mathbf{t})>0$, the ROI constraint can equivalently be written as
$p_i(\mathbf{t})/x_i(\mathbf{t})\leq v_i/r_i$, so $v_i/r_i$ is the largest
payment per unit of allocation the bidder can afford. This suggests an
equivalent representation of her type: define bidder $i$'s
\emph{unit-payment cap} as
\[
c_i = v_i / r_i,
\]
so that the ROI constraint can be rewritten as the
\emph{unit-payment-cap constraint}
\[
p_i(\mathbf{t}) \leq c_i x_i(\mathbf{t}).
\]
Since $r_i \in [1,\bar r]$, the unit-payment cap satisfies
$v_i/\bar r\leq c_i\leq v_i$.

The valuation--ROI representation $(v,r)$ and the valuation--cap representation
$(v,c)$ thus encode the same economic restriction, related by $c=v/r$.
Henceforth \(t=(v,c)\) is the default representation of a bidder's type,
with type space
\[
T=\left\{(v,c):0\leq v\leq\bar v,\ \frac{v}{\bar r}\leq c\leq v\right\};
\]
we use \((v,r)\) only when reasoning directly about the ROI constraint.
Figure~\ref{fig:type_space_transformation} illustrates the two
representations. Across the full type space, the unit-payment cap ranges
over $[0,\bar v]$; conditional on a fixed valuation $v$, its
range is $[v/\bar r,v]$. Utility is unchanged under the valuation--cap
representation, where affordability is simply
$p_i(\mathbf{t})\leq c_i x_i(\mathbf{t})$; accordingly, we use the terms
``ROI constraint'' and ``unit-payment-cap constraint'' interchangeably.

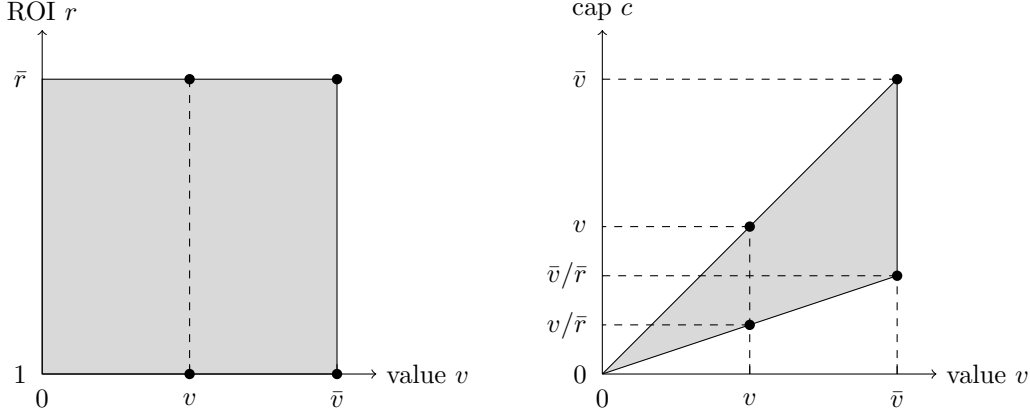
\begin{figure}[t]
\centering
\begin{tikzpicture}[scale=0.9, font=\small]
    \begin{scope}[shift={(0,0)}]
        \fill[gray!30] (0, 0) rectangle (3, 3);

        \draw (0,0) rectangle (3,3);

        \draw[->] (0,0) -- (3.4,0);
        \node at (1.7,-0.65) {value $v$};
        \draw[->] (0,0) -- (0,3.5) node[above] {ROI $r$};

        \node[below=2pt] at (0,0) {$0$};
        \node[left=2pt] at (0,0) {$1$};

        \draw[dashed] (1.5, 0) -- (1.5, 3);

        \fill (3, 0) circle (1.5pt);   
        \fill (1.5, 3) circle (1.5pt); 
        \fill (1.5, 0) circle (1.5pt); 
        \fill (3, 3) circle (1.5pt); 

        \draw (1.5, 0.1) -- (1.5, 0) node[below=3.5pt] {$v$};
        \draw (3, 0.1) -- (3, 0) node[below=3pt] {$\bar{v}$};

        \draw (0, 3) -- (0, 3) node[left=2pt] {$\bar r$};
    \end{scope}

    \begin{scope}[shift={(4.4,0)}]
        \fill[gray!30] (0,0) -- (3,3) -- (3,1) -- cycle;

        \draw (0,0) -- (3,3);
        \draw (0,0) -- (3,1);

        \draw (3,1) -- (3,3);

        \draw[->] (0,0) -- (3.4,0);
        \node at (1.7,-0.65) {value $v$};
        \draw[->] (0,0) -- (0,3.5) node[above] {cap $c$};


        \node[below=2pt] at (0,0) {$0$};
        \node[left=2pt] at (0,0) {$0$};

        \draw[dashed] (1.5, 0) -- (1.5, 1.5);
        \draw[dashed] (3, 0) -- (3, 1);

        \draw[dashed] (1.5, 0.5) -- (0, 0.5);
        \draw[dashed] (3, 1) -- (0, 1);
        \draw[dashed] (1.5, 1.5) -- (0, 1.5);
        \draw[dashed] (3, 3) -- (0, 3);

        \fill (3, 1) circle (1.5pt);     
        \fill (1.5, 1.5) circle (1.5pt); 
        \fill (1.5, 0.5) circle (1.5pt); 
        \fill (3, 3) circle (1.5pt); 

        \draw (1.5, 0.1) -- (1.5, 0) node[below=3.5pt] {$v$};
        \draw (3, 0.1) -- (3, 0) node[below=3pt] {$\bar{v}$};

        \draw (0, 0.5) -- (0, 0.5) node[left=2pt] {$v/\bar{r}$};
        \draw (0, 1) -- (0, 1) node[left=2pt] {$\bar{v}/\bar{r}$};
        \draw (0, 1.5) -- (0, 1.5) node[left=2pt] {$v$};
        \draw (0, 3) -- (0, 3) node[left=2pt] {$\bar{v}$};
    \end{scope}
\end{tikzpicture}
\caption{The grey area on the left illustrates the original type space, while that on the right illustrates the equivalent valuation--cap representation of type space. The points $(\bar{v}, \bar{r})$, $(v, 1)$, $(v, \bar{r})$, and $(\bar v, 1)$ in the left correspond to the points $(\bar{v}, \bar{v}/\bar{r})$, $(v, v)$,
$(v, v/\bar{r})$ and $(\bar v, \bar v)$ in the right, respectively.}
\Description{Two diagrams compare the valuation--ROI representation of the type space with the equivalent triangular valuation--cap representation.}
\label{fig:type_space_transformation}
\end{figure}

We study mechanisms that are \emph{truthful}, meaning both
dominant-strategy incentive compatible (DSIC) and individually rational
(IR). A mechanism is DSIC if truthful reporting is a dominant strategy
for every bidder:
\[
u(t_i,\mathbf{t}_{-i};t_i) \geq u(t',\mathbf{t}_{-i};t_i)
\quad \text{for all } t_i, t',\mathbf{t}_{-i},
\]
where $\mathbf{t}_{-i}$ denotes the reports of all bidders except $i$, and
$(t, \mathbf{t}_{-i})$ denotes the profile in which bidder $i$ reports type $t$.
A mechanism is IR if every bidder obtains non-negative utility under truthful reporting:
\[
u(t_i,\mathbf{t}_{-i};t_i) \geq 0 \quad \text{for all } t_i,\mathbf{t}_{-i}.
\]

For every bidder $i$ and type profile $\mathbf{t}$, we impose no positive
transfers, $p_i(\mathbf{t})\geq0$, and adopt the conventions
$x_i((0,0),\mathbf{t}_{-i})=\allowbreak p_i((0,0),\mathbf{t}_{-i})=0$
and $x_i(\mathbf{t})=0$ implies $p_i(\mathbf{t})=0$.

Truthfulness characterizations require no prior distribution.
For the revenue analysis, however, we make the following distributional
assumptions, with the public-information special cases given in
Section~\ref{sec:pricing-functions}. Bidder \(i\)'s valuation is drawn
from a publicly known distribution \(F_i\) with support \([0,\bar v]\)
and a density \(f_i\) that is positive on \((0,\bar v)\); her ROI
constraint is drawn independently from a publicly known distribution
\(G_i\) with support \([1,\bar r]\) and a density \(g_i\) that is positive
on \((1,\bar r)\). The pairs \((v_i,r_i)\) are independent across bidders.
The unit-payment cap \(c_i=v_i/r_i\) is not a separate primitive: its
distribution \(H_i\) is induced by \(F_i\) and \(G_i\), and the
transformed type distribution is the corresponding distribution of
\((v_i,c_i)\). Explicitly,
\[
H_i(c)
=
\Pr\!\left[\frac{v_i}{r_i}\le c\right]
=
\int_{1}^{\bar r}F_i(cr)g_i(r)\,dr,
\qquad c\in[0,\bar v].
\]
The density \(h_i\) of \(H_i\) is positive on \((0,\bar v)\), and the
induced caps are independent across bidders.
In single-bidder settings, we suppress the bidder index and write
\(F,G,H\) and \(f,g,h\).

For a mechanism $M$, we write $\operatorname{Rev}(M)$ for its expected
revenue when bidders' valuations and ROI constraints are drawn according
to the distributions specified above. We write $\mathrm{OPT}$ for the
supremum of expected revenue over all truthful mechanisms, and
$\mathrm{OPT}_{\mathrm{det}}$ for the supremum over deterministic truthful
mechanisms.

\section{Truthful Mechanisms}\label{sec:truthful-mechanisms}

We characterize truthfulness from the perspective of a single bidder,
studying the restrictions it imposes on her allocation and payment rules
as functions of her own report.
Unless otherwise stated, we hold the other bidders' reports fixed and
suppress the bidder index $i$ and $\mathbf{t}_{-i}$, writing allocations
and payments as $x(t)$ and $p(t)$, or equivalently $x(v,c)$ and $p(v,c)$.

We derive an allocation monotonicity condition \textbf{(A1)} and a
payment identity \textbf{(P)}, and then show that these, together with
the individual rationality condition \textbf{(IR)}, characterize
truthfulness. Along the way we derive further allocation properties that
give additional insight into the structure of truthful mechanisms.

\subsection{Allocation}
\label{sec:allocation}

Reporting truthfully to an IR mechanism is always safe, in that it never
violates the ROI constraint. Misreporting a lower cap is also always
safe: IR gives $p(v,c)\leq cx(v,c)$, so a bidder with cap $c'\geq c$
cannot violate her ROI constraint by reporting $(v,c)$. Misreporting a
higher cap may also be safe, but whether it is depends on the payment
assigned to that report. To make this dependence explicit, we introduce
the \emph{unit payment}: the payment per unit of allocation.

\begin{definition}
Given a mechanism $M=(x,p)$, the \emph{unit payment} for a type
$t=(v,c)$ is
\[
q_M(t)=q_M(v,c):=
\begin{cases}
p(t)/x(t), & \text{if }x(t)>0,\\
0, & \text{if }x(t)=0.
\end{cases}
\]
\end{definition}

\noindent
For simplicity, we drop the subscript $M$ and write $q(t)$ or $q(v,c)$.
Under our zero-allocation convention, $x(t)=0$ implies $p(t)=0$,
so $p(t)=q(t)x(t)$ for all $t\in T$; a mechanism can therefore be
represented equivalently by its allocation rule $x$ and unit-payment rule
$q$.

The distinction between a bidder's cap and the mechanism's unit payment
is central: a bidder with cap $c$ can safely report $t$ exactly when
$q(t)\le c$, and the report violates her ROI constraint otherwise. In
particular, IR implies $q(v,c)\le c$ for every type $(v,c)$.


The next lemma establishes a monotonicity relation between unit payments
and allocations in truthful mechanisms: a strict ordering of unit
payments forces the same strict ordering of allocations, and a weak
ordering of allocations forces the same weak ordering of unit payments.

\begin{lemma}\label{lem:unit-payment-monotonicity}
If a mechanism $(x,p)$ is truthful, then for any types $t$ and $t'$:
\begin{enumerate}
    \item if $q(t)<q(t')$, then $x(t)<x(t')$; and
    \item if $x(t)\leq x(t')$, then $q(t)\leq q(t')$.
\end{enumerate}
In particular, $x(t)=x(t')$ implies $q(t)=q(t')$.
\end{lemma}

Although unit-payment monotonicity is a useful structural property, it
does not by itself pin down allocation ordering in terms of the bidder's
valuation and unit-payment cap. We therefore combine comparisons of unit
payments with direct comparisons of valuations and caps, an approach that
reveals substantially more structure than either comparison alone.

\begin{theorem}\label{thm:Allocation}
If a mechanism $(x,p)$ is truthful,
then for any types $t=(v,c)$ and $t'=(v',c')$
\begin{description}
    \item[(A1)] if $v<v'$ and $q(t)\leq c'$, then $x(t)\leq x(t')$.
    \item[(A2)] if $c'<q(t)$, then $x(t')<x(t)$.
    \item[(A3)] if $v < v'$ and $c \leq c'$, then $x(v,c) \leq x(v',c')$.
\end{description}
\end{theorem}

\textbf{(A1)} and \textbf{(A2)} describe how affordability and valuation
comparisons constrain allocation ordering. The decisive comparison is
between one type's cap and the unit payment assigned to the other type.
If \(c' < q(t)\), type $t'$ cannot afford the outcome assigned to $t$,
so \(x(t') < x(t)\) regardless of valuation, as \textbf{(A2)} states.
When \(q(t) \leq c'\), that outcome \emph{is} affordable, and the
ordering follows standard valuation-based monotonicity: \(v < v'\)
implies \(x(t) \leq x(t')\), as captured by \textbf{(A1)}. \textbf{(A3)}
follows directly from \textbf{(A1)} and IR, since $q(t)\leq c$ and hence
$c\leq c'$ already guarantees affordability. It expresses allocation
monotonicity purely in terms of the types, without reference to unit
payments: a higher valuation together with a weakly higher cap implies a
weakly higher allocation.

It follows from \textbf{(A3)} that, holding the cap fixed ($c=c'$), a
bidder receives a weakly higher allocation when reporting a higher
valuation. This raises the natural question of how the allocation
behaves when the cap varies while the valuation stays fixed. Although
continuity of the allocation rule is not assumed, \textbf{(A3)} ensures
that $x(v,c)$ is nondecreasing in $v$ for every fixed $c$, and is
therefore continuous for almost every $v$. This almost-everywhere
continuity lets us show that, for any fixed pair of caps and almost
every valuation, the higher cap induces a weakly higher allocation, with
equality whenever the unit payment at the higher cap does not exceed the
lower cap. These are \textbf{(A4)} and \textbf{(A5)} below.

\begin{lemma}\label{lem:allocation_A_E}
If an IR mechanism $(x,p)$ satisfies \emph{\textbf{(A1)}}, then, for any
fixed caps $c$ and $c'$ and almost every $v$ such that $(v,c),(v,c')\in T$:
\begin{description}
    \item[(A4)] If $c<c'$, then $x(v,c)\leq x(v,c')$.
    \item[(A5)] If $q(v,c')\leq c<c'$, then $x(v,c)=x(v,c')$.
\end{description}
\end{lemma}

While \textbf{(A1)}--\textbf{(A3)} hold pointwise, \textbf{(A4)} and
\textbf{(A5)} hold only almost everywhere. Their measure-zero exceptions
are immaterial to our characterization of truthful mechanisms and to the
expected-revenue analysis, since they do not affect the relevant
integrals.

We close this section by emphasizing that the allocation properties
\textbf{(A1)}--\textbf{(A5)} are stated in terms of valuations and
unit-payment caps. Expressing types this way reveals monotonicity
properties of the allocation rule that are not apparent in the original
valuation--ROI representation, and this perspective is also essential for
the payment characterization that follows.

\subsection{Payment}\label{sec:payment}

We now show that the allocation rule uniquely determines the payment
rule in any truthful mechanism. Because allocation and payment depend on
both the valuation and the unit-payment cap, the standard one-dimensional
envelope argument does not apply directly, so we instead restrict the
mechanism to carefully chosen one-dimensional subsets of the type space.

For any $T'\subset T$, the \emph{restriction} of $(x,p)$ to $T'$ is the
mechanism $(x',p')$ with domain $T'$ such that $x'(t)=x(t)$ and
$p'(t)=p(t)$ for every $t\in T'$. Removing types and deviations preserves
truthfulness, since all incentive and participation inequalities on $T'$
are already required on $T$; hence every condition on allocation or
payment implied by truthfulness on $T'$ is necessary for a truthful
mechanism on the full type space.

Two subsets of the type space are of particular interest:
\[
T^= \;=\; \{(v,v)\in T : 0 \leq v \leq \bar v\},
\]
the set of \emph{diagonal types}, for which the valuation and unit-payment cap coincide
(equivalently, $r=1$). Types in $T\setminus T^=$ are \emph{off-diagonal}.
For each fixed cap $c$, we also consider the subset
\[
T^c \;=\; \left\{(v,c)\in T : c \leq v
\leq \min\{\bar r c,\bar v\}\right\},
\]
the set of types sharing that cap. Both subsets are one-dimensional: $T^=$
is identified with $[0,\bar v]$, and $T^c$ with
$[c,\min\{\bar r c,\bar v\}]$.

On both subsets, the restricted mechanism admits the usual quasilinear
incentive analysis, but for different reasons. On $T^c$, all types share
the same cap $c$, and IR gives $p(v',c)\leq c x(v',c)$, so every deviation within this subset is
affordable. On $T^=$, the cap equals the valuation, so the ROI
constraint coincides with the standard IR requirement that payment not
exceed the value obtained. Standard quasilinear DSIC therefore applies:
affordable deviations are ruled out by DSIC, and unaffordable deviations
yield negative quasilinear utility and so cannot improve on truthful IR
utility. Applying the usual single-dimensional envelope theorem, we
first deduce $p(c,c)$ along $T^=$ and then deduce $p(v,c)$ relative to
$p(c,c)$ along $T^c$, which yields the following payment identity.

\begin{theorem}\label{thm:Payment}
If a mechanism $(x,p)$ is truthful, then
\begin{equation*}
    \emph{\textbf{(P)}} \;\; p(v, c) = v x(v,c) - \int_{0}^{c} x(z,z)\, dz - \int_{c}^{v} x(z,c)\, dz.
\end{equation*}
\end{theorem}

Condition \textbf{(P)} is a closed-form generalization of Myerson's
payment identity to the two-dimensional type space: it integrates the
allocation first along the diagonal from $(0,0)$ to $(c,c)$ and then
along the fixed-cap slice from $(c,c)$ to $(v,c)$, so fixing the
allocation uniquely fixes every payment. When $v=c$, the second integral
vanishes and \textbf{(P)} reduces to the classical Myerson identity.

\subsection{Characterization}\label{sec:characterisation}

The preceding subsections establish necessary conditions for truthfulness.
For sufficiency, we must also require the payments induced by
\textbf{(P)} to respect every type's ROI constraint. In the unit-payment
representation, this requirement takes the simple form
\[
    \textbf{(IR)} \quad q(v,c) \leq c \quad \text{for all } (v,c) \in T.
\]
Since $p(t)=q(t)x(t)$, the ROI constraint $p(t)\leq cx(t)$ holds if
and only if $q(t)\leq c$, so \textbf{(IR)} is both necessary and
sufficient for ROI compliance. It follows that \textbf{(A1)},
\textbf{(P)}, and \textbf{(IR)} are together necessary and sufficient
for a mechanism to be truthful.

\begin{theorem}\label{thm:characterisation}
A mechanism $(x,p)$ is truthful if and only if it satisfies
\emph{\textbf{(A1)}}, \emph{\textbf{(P)}} and \emph{\textbf{(IR)}}.
\end{theorem}

\noindent
While Theorem~\ref{thm:Allocation} and Lemma~\ref{lem:allocation_A_E}
establish \textbf{(A1)}--\textbf{(A5)} as necessary properties,
Theorem~\ref{thm:characterisation} gives the exact characterization of
truthful mechanisms. In particular, \textbf{(A1)} and \textbf{(IR)}
directly imply \textbf{(A3)}, and, as Lemma~\ref{lem:allocation_A_E}
shows, imply \textbf{(A4)} and \textbf{(A5)} almost everywhere.
Condition \textbf{(A2)}, by contrast, does not follow from \textbf{(A1)}
and \textbf{(IR)} alone. It follows from the full characterization,
including the payment identity \textbf{(P)}, and therefore need not be
imposed separately.

Theorem~\ref{thm:characterisation} gives an exact test for truthfulness
but does not by itself provide a way to construct a truthful mechanism.
Even when an allocation has the monotonicity structure described by
\textbf{(A1)}--\textbf{(A5)}, the payment rule determined by \textbf{(P)}
may induce a unit payment exceeding the bidder's cap and hence violate
\textbf{(IR)}. Unlike in the standard quasilinear setting, the payment
identity does not make individual rationality automatic here; the
allocation and its induced payments must instead be designed jointly.
This difficulty motivates our study of two restricted but practically
important settings: deterministic mechanisms and single-bidder pricing
functions.

\subsection{Diagonal Determination of Interior Allocations}
\label{sec:diagonal-determination}

Under continuity of the allocation on the diagonal type space $T^{=}$,
the allocation rule admits a sharper structural description. The next
lemma complements the truthful-mechanism characterization in Theorem~\ref{thm:characterisation} with an operational characterization: the one-dimensional diagonal
allocation determines every
interior off-diagonal allocation by selecting the largest diagonal
allocation whose unit payment is affordable. As throughout this section,
we fix a bidder and the reports of all other bidders and suppress them from
the notation.

\begin{lemma}\label{lem:diagonal-determination}
Let $(x,p)$ be a DSIC and IR mechanism, and suppose that the restriction of
the allocation rule to $T^{=}$ is continuous; that is, $x(z,z)$ is
continuous in $z$ on $[0,\bar v]$.
Then, for every type $t=(v,c)\in T\setminus T^{=}$ with $v<\bar v$,
\begin{equation}\label{eq:diagonal-determination}
x(v,c)
=
\sup\bigl\{x(z,z):z\in[c,v]\text{ and }q(z,z)\le c\bigr\}.
\end{equation}
In particular, every such type receives the same allocation as some
diagonal type.
\end{lemma}

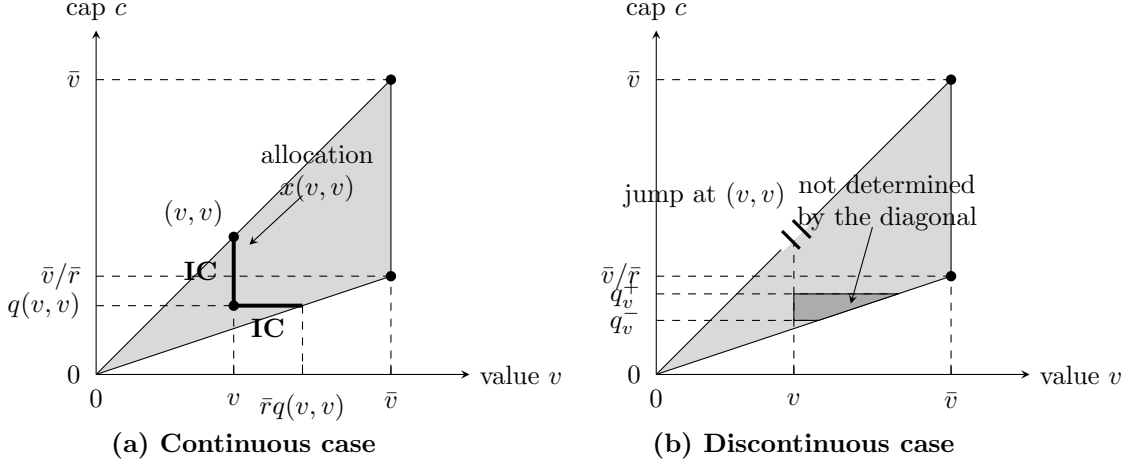
\begin{figure}[t]
\centering
\begin{tikzpicture}[scale=1.3,>=stealth,font=\small]
    \begin{scope}
        \fill[gray!30] (0,0) -- (3,3) -- (3,1) -- cycle;
        \draw (0,0) -- (3,3);
        \draw (0,0) -- (3,1);
        \draw (3,1) -- (3,3);
        \draw[->] (0,0) -- (3.8,0) node[right] {value $v$};
        \draw[->] (0,0) -- (0,3.5) node[above] {cap $c$};

        \node[below=2pt] at (0,0) {$0$};
        \node[left=2pt] at (0,0) {$0$};
        \coordinate (D) at (1.4,1.4);
        \coordinate (Q) at (1.4,0.7);
        \coordinate (R) at (2.1,0.7);
        \draw[dashed] (1.4,0) -- (D);
        \draw[dashed] (0,0.7) -- (Q);
        \draw[dashed] (R) -- (2.1,0);
        \draw[dashed] (3,0) -- (3,1);
        \draw[dashed] (0,1) -- (3,1);
        \draw[dashed] (0,3) -- (3,3);

        \draw[line width=1.5pt] (D) -- (Q) -- (R);
        \fill (D) circle (1.5pt);
        \fill (Q) circle (1.5pt);
        \fill (3,1) circle (1.5pt);
        \fill (3,3) circle (1.5pt);

        \draw (1.4,0.1) -- (1.4,0) node[below=3.5pt] {$v$};
        \draw (2.1,0.1) -- (2.1,0)
            node[below=3.5pt] {$\bar r q(v,v)$};
        \draw (3,0.1) -- (3,0) node[below=3pt] {$\bar v$};
        \node[left=2pt] at (0,0.7) {$q(v,v)$};
        \node[left=2pt] at (0,1) {$\bar v/\bar r$};
        \node[left=2pt] at (0,3) {$\bar v$};

        \node[above left] at (D) {$(v,v)$};
        \node[left] at (1.34,1.05) {\textbf{IC}};
        \node[below] at (1.75,0.67) {\textbf{IC}};
        \node[align=center] at (2.25,2.05) {allocation\\$x(v,v)$};
        \draw[->] (2.1,1.82) -- (1.56,1.32);
        \node at (1.5,-0.72) {\textbf{(a) Continuous case}};
    \end{scope}

    \begin{scope}[shift={(5.7,0)}]
        \fill[gray!30] (0,0) -- (3,3) -- (3,1) -- cycle;
        \draw (0,0) -- (3,3);
        \draw (0,0) -- (3,1);
        \draw (3,1) -- (3,3);
        \draw[->] (0,0) -- (3.8,0) node[right] {value $v$};
        \draw[->] (0,0) -- (0,3.5) node[above] {cap $c$};

        \node[below=2pt] at (0,0) {$0$};
        \node[left=2pt] at (0,0) {$0$};
        \coordinate (J) at (1.4,1.4);
        \fill[gray!65]
            (1.4,0.55) -- (1.65,0.55) -- (2.46,0.82)
            -- (1.4,0.82) -- cycle;
        \draw
            (1.4,0.55) -- (1.65,0.55) -- (2.46,0.82)
            -- (1.4,0.82) -- cycle;
        \draw[dashed] (1.4,0) -- (J);
        \draw[dashed] (0,0.55) -- (1.65,0.55);
        \draw[dashed] (0,0.82) -- (2.46,0.82);
        \draw[dashed] (3,0) -- (3,1);
        \draw[dashed] (0,1) -- (3,1);
        \draw[dashed] (0,3) -- (3,3);

        \draw[line width=2pt,white] (1.27,1.27) -- (1.53,1.53);
        \draw[very thick] (1.28,1.46) -- (1.46,1.28);
        \draw[very thick] (1.39,1.57) -- (1.57,1.39);
        \node[above left,align=right] at (1.42,1.58)
            {jump at $(v,v)$};

        \fill (3,1) circle (1.5pt);
        \fill (3,3) circle (1.5pt);
        \draw (1.4,0.1) -- (1.4,0) node[below=3.5pt] {$v$};
        \draw (3,0.1) -- (3,0) node[below=3pt] {$\bar v$};
        \node[left=2pt] at (0,0.55) {$q_v^-$};
        \node[left=2pt] at (0,0.82) {$q_v^+$};
        \node[left=2pt] at (0,1) {$\bar v/\bar r$};
        \node[left=2pt] at (0,3) {$\bar v$};

        \node[align=center] at (2.35,1.75)
            {not determined\\by the diagonal};
        \draw[->] (2.2,1.5) -- (1.98,0.72);
        \node at (1.5,-0.72) {\textbf{(b) Discontinuous case}};
    \end{scope}
\end{tikzpicture}
\caption{How the diagonal allocation propagates into the interior of the
type space. \textbf{(a)} Under diagonal continuity, at an interior valuation $v$, all types on the highlighted $L$-shape receive allocation
$x(v,v)$: incentive compatibility gives equality down to cap $q(v,v)$ and
extends it horizontally along that cap, away from the boundary $v=\bar v$. The illustration
assumes $q(v,v)\ge v/\bar r$ so that the
bottom of the $L$ lies in the type space. \textbf{(b)} At a jump
at $v$, let $q_v^-<q_v^+$ denote the limiting unit payments associated with
the lower and upper diagonal allocation levels. The diagonal allocation
need not determine allocations in the shaded region
$\{(v',c):q_v^-<c<q_v^+,\ v\le v'\le\min\{\bar r c,\bar v\}\}$.}
\Description{Two diagrams of the valuation--unit-payment-cap type space.
The first highlights an L-shaped set extending downward and rightward from
the diagonal type $(v,v)$ at a continuity point. The second shades a
region whose allocation is not determined when the diagonal allocation
has a jump at $v$.}
\label{fig:diagonal-determination}
\end{figure}

Figure~\ref{fig:diagonal-determination} illustrates the consequence of
Lemma~\ref{lem:diagonal-determination}: under diagonal continuity,
specifying the monotone allocation rule $z\mapsto x(z,z)$ over diagonal
types determines all off-diagonal allocations except those with
$v=\bar v$, for which there is no higher diagonal type available for
comparison. Thus, as in the single-dimensional setting, the diagonal allocation of a truthful mechanism determines its
interior allocation rule, since
\textbf{(P)} then uniquely determines all payments. The potentially
undetermined types lie on the boundary
$\{(\bar v,c):\bar v/\bar r\le c<\bar v\}$. This boundary has Lebesgue
measure zero and therefore does not affect expected-revenue calculations
that integrate over types with a density.

\section{Deterministic Mechanisms}\label{sec:deterministic-mechanisms}

We now examine how the general characterization specializes to the
deterministic setting, where the allocation rule is restricted to binary
outcomes. This restriction simplifies the structure of truthful
mechanisms and leads to a characterization that mirrors the classical
single-dimensional case. We then derive a tight upper bound on the
expected revenue achievable by deterministic truthful mechanisms, and a
family of $\sigma$-increment mechanisms whose expected revenue approaches
this bound, giving a $\bar r$-approximation, in the limit, to the optimal
expected revenue among all truthful mechanisms.

\subsection{Characterization of Truthfulness}
\label{sec:deterministic-thresholds}

The general characterization specializes to the following conditions
under determinism.

\begin{theorem}\label{thm:characterisation-deterministic}
A deterministic mechanism $(x,p)$ is truthful if and only if it
satisfies
\begin{description}
    \item[(DA1)] if $c<c'$, then $x(v,c)\leq x(v',c')$;
    \item[(DA2)] $x(v,c)=x(v',c)$ for all $v,v'>c$;
    \item[(DA3)] if
    $\kappa:=\inf\{c:x(v,c)=1\text{ for some }v>c\}$ exists in
    $[0,\bar v]$, then $x(v,\kappa)=1$ for every $v>\kappa$;
    \item[(DP)]
    $p(v,c)=c x(v,c)-\int_{0}^{c}x(z,z)\,dz$.
\end{description}
\end{theorem}

\textbf{(DA1)} imposes monotonicity in the unit-payment cap, while
\textbf{(DA2)} implies that the allocation of an off-diagonal type
depends only on its cap. Since the allocation is deterministic, these
conditions together imply a threshold $\kappa$: every off-diagonal type
with $c<\kappa$ loses, and every off-diagonal type with $c>\kappa$ wins.
\textbf{(DA3)} additionally requires every off-diagonal type $(v,\kappa)$
with $v>\kappa$ to win, though the diagonal type $(\kappa,\kappa)$ may
either win or lose. Finally, \textbf{(DP)} uniquely determines payments
from the allocation.

\subsection{Revenue Maximization}
\label{sec:deterministic-revenue}

We now analyze the revenue of deterministic truthful mechanisms, writing
the bidder index and the other bidders' reports explicitly. The
deterministic characterization makes the unit-payment cap the effective
one-dimensional parameter: by \textbf{(DA2)}, the allocation of an
off-diagonal type depends only on its cap, and \textbf{(DP)} then implies
the same for its payment. We therefore identify an off-diagonal type
with its cap and write bidder $i$'s allocation and payment as
$x_i(c;\mathbf t_{-i})$ and $p_i(c;\mathbf t_{-i})$. Under this
convention, \textbf{(DP)} gives
\[
p_i(c;\mathbf t_{-i})
=c x_i(c;\mathbf t_{-i})
-\int_0^c x_i(z;\mathbf t_{-i})\,dz.
\]
Define bidder $i$'s
\emph{virtual unit-payment cap}, or simply \emph{virtual cap}, by
\[
\varphi_i(c)=c-\frac{1-H_i(c)}{h_i(c)},
\qquad c\in(0,\bar v).
\]
We call $H_i$ \emph{regular} if $\varphi_i$ is nondecreasing on
$(0,\bar v)$. Since $h_i$ may vanish at the endpoints, the formula above
need not be defined there; we therefore set $\varphi_i(0)=0$ and
$\varphi_i(\bar v)=\bar v$ to define the mechanism at every reported type profile.
These endpoint conventions do not affect expected revenue because the
cap distributions are atomless.
Conditional on $\mathbf t_{-i}$, bidder $i$'s expected payment is
\[
\begin{aligned}
\mathbb E_{c_i}[p_i(c_i;\mathbf t_{-i})]
&=\int_0^{\bar v}
\left(c x_i(c;\mathbf t_{-i})
-\int_0^c x_i(z;\mathbf t_{-i})\,dz\right)h_i(c)\,dc\\
&=\int_0^{\bar v}
\left[c h_i(c)-(1-H_i(c))\right]
x_i(c;\mathbf t_{-i})\,dc\\
&=\mathbb E_{c_i}\!\left[
\varphi_i(c_i)x_i(c_i;\mathbf t_{-i})\right].
\end{aligned}
\]
The second equality exchanges the order of integration and uses
$\int_z^{\bar v}h_i(c)\,dc=1-H_i(z)$; the last uses the definition of
$\varphi_i$.
Averaging over the other bidders' reports and summing over bidders therefore
yields
\[
\mathbb E\!\left[\sum_{i=1}^n p_i(\mathbf t)\right]
=\mathbb E\!\left[\sum_{i=1}^n
\varphi_i(c_i)x_i(\mathbf t)\right].
\]
This identity mirrors Myerson's virtual-value derivation, with
unit-payment caps taking the role of values, and it therefore suggests
applying Myerson's optimal single-item mechanism directly to the
reported caps. Assume now that $H_i$ is regular for every $i$, and call
the resulting mechanism the \emph{Myerson-cap mechanism}, denoted
$M^0=(x^0,p^0)$: it allocates the item to a bidder with the highest
positive virtual cap, breaking ties by a fixed bidder-priority order
independent of the reports. This allocation is nondecreasing in the cap
of each bidder. For every bidder $i$ and profile $\mathbf t_{-i}$, define
the winning threshold
\[
\kappa_i(\mathbf t_{-i})
:=
\inf\{c_i\in[0,\bar v]:x_i^0(t_i,\mathbf t_{-i})=1
\text{ for some }t_i=(v_i,c_i)\in T\},
\]
where the infimum is $\infty$ if bidder $i$ can never win. A winner pays
$\kappa_i(\mathbf t_{-i})$, and a loser pays zero.
Since $M^0$ maximizes $\sum_i\varphi_i(c_i)x_i(\mathbf t)$ pointwise,
the revenue identity gives
$\operatorname{Rev}(M^0)\geq\mathrm{OPT}_{\mathrm{det}}$, where
$\operatorname{Rev}(M^0)$ is evaluated under truthful reporting.

The Myerson-cap mechanism, however, need not be DSIC. Fix $i$ and
$\mathbf t_{-i}$, and write $\kappa_i=\kappa_i(\mathbf t_{-i})$. Suppose
bidder $i$ with $c_i=\kappa_i$ ties with another bidder in virtual cap
and loses under the fixed tie-breaking rule. The off-diagonal type
$(v_i,\kappa_i)$ then loses under truthful reporting, but by instead
reporting any cap $c_i'\in(\kappa_i,v_i)$, she wins, pays $\kappa_i$, and
obtains positive utility. Thus $M^0$ fails DSIC at this type profile.
Such a failure, however, requires $c_i=\kappa_i(\mathbf t_{-i})$, a
measure-zero event under the atomless cap distribution. So although
$M^0$ need not be DSIC pointwise, it is DSIC almost surely. Allowing the
affected types to take this profitable deviation leaves expected revenue
unchanged relative to truthful reporting; nevertheless, DSIC must hold at every type profile.

To restore DSIC without materially changing the Myerson-cap mechanism,
we can introduce a small gap above the Myerson threshold: a bidder wins only
once her cap reaches $\kappa_i(\mathbf t_{-i})+\sigma$, and then pays
this shifted threshold. In the failure case above, the bidder's true cap
is $c_i=\kappa_i(\mathbf t_{-i})$; under the shifted mechanism, any
report that makes her win requires payment
$\kappa_i(\mathbf t_{-i})+\sigma>c_i$, which is unaffordable, so the
profitable deviation is eliminated. The allocation changes only for
bidders with caps in
$[\kappa_i(\mathbf t_{-i}),\kappa_i(\mathbf t_{-i})+\sigma)$: they may
win under the Myerson-cap mechanism but lose under the shifted one, and
every bidder who still wins pays an additional $\sigma$. This motivates
the following mechanism.

\begin{definition}[$\sigma$-increment mechanism]\label{def:sigma-increment}
For $\sigma>0$, the \emph{$\sigma$-increment mechanism}
$M^\sigma=(x^\sigma,p^\sigma)$ is defined by
\[
x_i^\sigma(\mathbf t)=1
\quad\text{if and only if}\quad
c_i\geq \kappa_i(\mathbf t_{-i})+\sigma,
\]
and
\[
p_i^\sigma(\mathbf t)
=
\begin{cases}
\kappa_i(\mathbf t_{-i})+\sigma,
&\text{if }x_i^\sigma(\mathbf t)=1,\\
0,&\text{otherwise}.
\end{cases}
\]
If no bidder satisfies the allocation inequality, the item is not
allocated.
\end{definition}

Under regularity, the resulting mechanism is truthful for every
$\sigma>0$, and its revenue converges to that of the Myerson-cap
mechanism as $\sigma$ vanishes.

\begin{theorem}\label{thm:optimal-deterministic}
Let $H_i$ be regular for every $i$. For any $\sigma>0$,
$M^\sigma$ is a deterministic truthful mechanism. Moreover,
\[
\lim_{\sigma\downarrow0}\operatorname{Rev}(M^\sigma)
=\operatorname{Rev}(M^0).
\]
\end{theorem}

Because every $M^\sigma$ is deterministic and truthful,
$\operatorname{Rev}(M^\sigma)\leq\mathrm{OPT}_{\mathrm{det}}$.
Letting $\sigma$ tend to zero in Theorem~\ref{thm:optimal-deterministic} gives
$\operatorname{Rev}(M^0)\leq\mathrm{OPT}_{\mathrm{det}}$, and combined
with the earlier upper bound this gives the reverse inequality, so
$\mathrm{OPT}_{\mathrm{det}}=\operatorname{Rev}(M^0)$. In particular, for
every $\varepsilon>0$, some $\sigma>0$ satisfies
$\operatorname{Rev}(M^\sigma)\geq\operatorname{Rev}(M^0)-\varepsilon$.

\subsection{Revenue Gap between Deterministic and Randomized Mechanisms}
\label{sec:deterministic-approximation}

We next compare $\mathrm{OPT}_{\mathrm{det}}$ with $\mathrm{OPT}$ through
three lemmas. The first gives a single-dimensional upper bound on
payments along any fixed-ROI slice.

\begin{lemma}
\label{lem:fixed-roi-payment-bound}
Let $(x,p)$ be a truthful mechanism.
For every $r\in[1,\bar r]$, let $(x_r,p_r)$ be its restriction to
$\{(v,v/r):v\in[0,\bar v]\}$, parameterized by $v$.
Then, for every $v\in[0,\bar v]$,
\[
p_r(v)
\le
v x_r(v)-\int_0^v x_r(z)\,dz.
\]
\end{lemma}
By \textbf{(A3)}, $x_r$ is nondecreasing in the valuation along each
fixed-ROI slice, so the right-hand side of
the above inequality is exactly the classical Myerson
payment for this allocation rule. Keeping $x_r$ and replacing $p_r$ by
that right-hand side therefore gives a truthful mechanism in the
ordinary quasilinear setting with weakly higher payments for every
valuation.

To compare revenues, consider the setting obtained by fixing $r_i=1$ for
every bidder while keeping the valuation distributions
$F_1,\ldots,F_n$ unchanged. Since $c_i=v_i$, this is exactly the
ordinary quasilinear single-item auction setting; let $\mathrm{OPT}_{r=1}$
denote its optimal expected revenue.

The next lemma uses the fixed-ROI payment bound to show that
$\mathrm{OPT}$ cannot exceed $\mathrm{OPT}_{r=1}$.

\begin{lemma}
\label{lem:opt-upper-value}
$\mathrm{OPT}\le \mathrm{OPT}_{r=1}$.
\end{lemma}

The complementary comparison goes the other way. Since every cap
satisfies $c_i=v_i/r_i\ge v_i/\bar r$, the cap distributions are at
least as large as the valuation distributions scaled by $1/\bar r$, and
the next lemma shows that $\mathrm{OPT}_{\mathrm{det}}$ is accordingly at
least a $1/\bar r$ fraction of $\mathrm{OPT}_{r=1}$.

\begin{lemma}
\label{lem:det-lower-value}
Suppose $H_i$ is regular for every $i$. Then
\[
\frac{1}{\bar r}\,\mathrm{OPT}_{r=1}
\le \mathrm{OPT}_{\mathrm{det}}.
\]
\end{lemma}

Lemma~\ref{lem:opt-upper-value} upper-bounds $\mathrm{OPT}$ by
$\mathrm{OPT}_{r=1}$, while Lemma~\ref{lem:det-lower-value} guarantees
that $\mathrm{OPT}_{\mathrm{det}}$ captures a $1/\bar r$ fraction of
$\mathrm{OPT}_{r=1}$; combining them gives
\[
\mathrm{OPT}_{\mathrm{det}}
\ge
\frac{1}{\bar r}\,\mathrm{OPT}.
\]
By Theorem~\ref{thm:optimal-deterministic}, the revenues of the
deterministic truthful mechanisms $M^\sigma$ approach
$\mathrm{OPT}_{\mathrm{det}}$ as $\sigma$ tends to zero, which yields the
following revenue approximation result.

\begin{theorem}
\label{thm:deterministic-approximation}
Suppose $H_i$ is regular for every $i$.
For every $\varepsilon>0$, there exists a deterministic truthful mechanism
$M$ satisfying
\[
\operatorname{Rev}(M)\geq\frac{1}{\bar r}\,\mathrm{OPT}-\varepsilon.
\]
\end{theorem}

For each $\varepsilon>0$, the $\sigma$-increment mechanism $M^\sigma$
achieves this bound for all sufficiently small $\sigma>0$.

Although we have stated the revenue results under regularity for
simplicity, all results in this section extend to non-regular cap
distributions. Apply Myerson's \emph{ironing technique} to the virtual
caps and break ties using a fixed ordering of bidders. The resulting
ironed Myerson-cap mechanism is deterministic and has a nondecreasing
allocation rule in each cap. Using its winning thresholds in the
$\sigma$-increment construction preserves feasibility and truthfulness,
and atomlessness ensures that expected revenue converges to the optimal
quasilinear cap revenue as $\sigma$ tends to zero. Thus both asymptotic
optimality among deterministic truthful mechanisms and the
$1/\bar r$ revenue approximation continue to hold without regularity.

\section{Pricing Functions}
\label{sec:pricing-functions}



In the single-bidder setting, we study pricing functions.
A \emph{pricing function} is a continuous function
\[
\pi:[0,1]\to\mathbb{R}_{+},
\]
with \(\pi(0)=0\), where \(\pi(x)\) is the payment for choosing quantity
\(x\). A bidder with type \((v,c)\) chooses an affordable quantity to
maximize her utility \(vx-\pi(x)\). Thus the chosen quantity lies in
\[
\argmax_{x\in[0,1]}\{vx-\pi(x):\pi(x)\le cx\}.
\]
Continuity and \(\pi(0)=0\) ensure that the affordable set is nonempty
and compact, so a largest utility maximizer exists. Throughout this
section, the bidder breaks ties by choosing the largest affordable
quantity among those that maximize utility, except that type \((0,0)\)
receives zero allocation and pays zero. We denote the resulting
allocation by \(x_{\pi}(v,c)\) and the payment by
\(p_{\pi}(v,c)=\pi(x_{\pi}(v,c))\).

The mechanism induced by a pricing function is truthful by construction, since
the bidder selects her preferred affordable quantity. We show that
every truthful mechanism can be replaced by a convex and nondecreasing
pricing function without reducing any type's payment. Hence, it
suffices to consider such pricing functions for revenue maximization
in the single-bidder setting.

\begin{theorem}
\label{thm:pricing-reduction}
Let $S\subseteq T$, and let $(x,p)$ be a truthful
mechanism on $S$. Then there exists a continuous, convex, and
nondecreasing pricing function $\pi$, with $\pi(0)=0$ and
$\pi(1)\le\bar v$, such that
\[
p_{\pi}(t)\ge p(t)
\]
for every $t\in S$.
\end{theorem}

The proof starts from the partial menu of allocation--payment pairs
offered by the original mechanism and takes its lower convex envelope.
Although this operation can lower the price at a fixed quantity, bidders
respond by choosing weakly larger quantities. If an original allocation--payment pair
remains on the envelope, the bidder can obtain at least the same
quantity at a weakly higher payment, and if the envelope gives a strictly
lower price at the original quantity, truthfulness makes the bidder prefer moving
along the new linear segment until she reaches a quantity with weakly
higher payment or her affordability boundary. Either way, the
payment does not decrease.

For the remainder of this section, we restrict attention to pricing
functions \(\pi\) with the properties specified in
Theorem~\ref{thm:pricing-reduction}.
For such a pricing function, convexity and \(\pi(0)=0\) imply that the
average price \(\pi(x)/x\) is nondecreasing. At every
\(x\in(0,1)\) where \(\pi\) is differentiable, a type \((v,c)\) with
\(v>0\) purchases at least \(x\) if and only if
\[
v\ge \pi'(x)
\ \text{ and }\
c\ge \frac{\pi(x)}{x}.
\]
The first inequality says the bidder's valuation is at least the
marginal price at \(x\), so her marginal utility \(v-\pi'(x)\) is
nonnegative---otherwise reducing the quantity slightly below \(x\)
would increase her utility. The second inequality ensures that quantity
\(x\) is affordable under her unit-payment-cap constraint.

The optimal pricing function for the fully private setting remains unresolved; the next two subsections instead characterize optimal pricing when either the valuation or the ROI constraint is public.

\subsection{Optimal Pricing under Public Valuation}
\label{sec:public-valuation}

We now study the setting in which the valuation \(v\) is publicly known
and the ROI constraint \(r\) remains private. A bidder's type \((v,c)\)
can therefore be identified by the unit-payment cap \(c\) alone, and we
write the type simply as \(c\), suppressing the fixed valuation in
\(x_\pi(c)\) and \(p_\pi(c)\). Since the ROI constraint lies in
\([1,\bar r]\), we have \(c\in[\underline c,v]\), where
\(\underline c:=v/\bar r\). Here \(H\) and \(h\) denote the cap
distribution and density induced by \(G\) at this fixed valuation, with
support \([\underline c,v]\) and \(h>0\) on \((\underline c,v)\).

All types rank quantities by the same utility expression
\(vx-\pi(x)\), while the cap \(c\) determines affordability. For an
additional quantity \(\Delta x\), increasing the payment by
\(v\Delta x\) extracts the additional value without changing utility,
provided the larger quantity remains affordable. This motivates
considering a linear pricing segment with slope \(v\).

To make this intuition precise, we first normalize an arbitrary pricing
function so that full quantity maximizes utility when the ROI constraint
is ignored.
The following lemma shows that this normalization does not reduce any
type's payment.

\begin{lemma}\label{lem:public-valuation-normalization}
Fix \(v\in(0,\bar v)\). For every continuous, convex, and nondecreasing
pricing function \(\pi\), with \(\pi(0)=0\) and \(\pi(1)\le\bar v\),
there exists a pricing function \(\widetilde{\pi}\) with the same
properties such that
\[
1\in\argmax_{x\in[0,1]}\{vx-\widetilde{\pi}(x)\}
\]
and
\[
p_{\widetilde{\pi}}(v,c)\ge p_\pi(v,c)
\qquad
\text{for every }c\in[\underline c,v].
\]
\end{lemma}

For a normalized pricing function \(\pi\), write \(z=\pi(1)\). The normalization gives \(v-z\ge vx-\pi(x)\), which, together with
\(\pi(x)\ge0\), implies the price lower bound
\[
\pi(x)\ge \bigl[vx-(v-z)\bigr]_+,
\]
where \([a]_+:=\max\{a,0\}\). Comparison with zero quantity gives
\(z\le v\). If \(z<\underline c\), every type buys full quantity at
price \(z\), so raising \(z\) to \(\underline c\) improves revenue; it
therefore suffices to consider \(z\in[\underline c,v]\).

For a type \(c<z\), combining the price lower bound with affordability
\(p_\pi(c)\le cx_\pi(c)\) yields
\[
p_\pi(c)\le\frac{c(v-z)}{v-c}.
\]
For \(c\ge z\), the bidder chooses full quantity and pays \(z\).

This suggests setting the price equal to the payment lower bound,
giving the pricing function
\[
\pi_z(x)=[vx-(v-z)]_+.
\]
The price is zero up to quantity \((v-z)/v\) and then rises with slope
\(v\), so utility increases up to this quantity and remains constant
thereafter. The bidder
chooses the largest affordable quantity, which is full quantity at payment
\(z\) if \(c\ge z\). If \(c<z\), the chosen quantity \(x\) satisfies
\(vx-(v-z)=cx\), giving
\[
x=\frac{v-z}{v-c},
\qquad
\pi_z(x)=\frac{c(v-z)}{v-c}.
\]
Thus \(\pi_z\)
attains the payment upper bound \(c(v-z)/(v-c)\) for every \(c<z\)
and preserves the payment \(z\) for every \(c\ge z\). It therefore
maximizes revenue among normalized pricing functions with the same
price \(z\) for full quantity.

It remains to optimize over \(z\). The corresponding expected revenue is
\[
R(z)
=
\int_{\underline c}^{z}
\frac{c(v-z)}{v-c}h(c)\,dc
+z\bigl[1-H(z)\bigr].
\]
The integral is the expected payment from types with \(c<z\), while the
second term is the payment \(z\) from types with \(c\ge z\). Increasing
\(z\) raises the payment from types that can afford full quantity, but
reduces the quantity and hence the payment of types with lower
caps. For the given
\(H\) and \(h\), \(R\) is strictly concave, and its first-order
condition determines the unique revenue-optimal choice \(c^*\).
Substituting \(c^*\) into \(\pi_z\) gives the optimal pricing
function stated below.

\begin{theorem}\label{thm:public-valuation-optimal-pricing}
Consider a single bidder with publicly known valuation
\(v\in(0,\bar v)\) and private ROI constraint \(r\).
An optimal pricing function is
\[
\pi^*(x)=\bigl[vx-(v-c^*)\bigr]_+,
\]
where $c^*\in(\underline c,v)$ is the unique solution of
\[
\int_{\underline c}^{c^*}\frac{c}{v-c}h(c)\,dc
=
1-H(c^*).
\]
\end{theorem}

Figure~\ref{fig:public-value-pricing} illustrates the pricing function.
The price is zero for quantities up to \((v-c^*)/v\) and strictly
positive above this cutoff.
Under \(\pi^*\), a type \(c<c^*\) chooses quantity
\((v-c^*)/(v-c)\), while a type
\(c\ge c^*\) chooses the full quantity.
Both chosen quantities lie strictly above the cutoff, so every type makes a strictly positive
payment despite the zero-price segment.

\begin{figure}[t]
\centering
\begin{tikzpicture}[x=5.1cm,y=2.35cm]
  \draw[->] (0,0) -- (1.12,0) node[right] {quantity \(x\)};
  \draw[->] (0,0) -- (0,1.16) node[above] {price \(\pi^*(x)\)};

  \draw[very thick] (0,0) -- (0.38,0) -- (1,1);
  \draw[dashed] (1,0) -- (1,1);
  \draw[dashed] (0,1) -- (1,1);

  \draw (0.38,0.025) -- (0.38,-0.025)
    node[below] {\(\frac{v-c^*}{v}\)};
  \draw (1,0.025) -- (1,-0.025)
    node[below] {\(1\)};
  \draw (0.008,1) -- (-0.008,1)
    node[left] {\(c^*\)};

  \node[above] at (0.19,0.02) {zero price};
  \node[above left] at (0.78,0.66) {slope \(v\)};
\end{tikzpicture}
\caption{Optimal pricing with public valuation. Quantities up to
\((v-c^*)/v\) have zero price, after which the price rises linearly with slope \(v\) to the
full-quantity price \(c^*\).}
\Description{A pricing function that is zero for small allocation quantities and then rises linearly to the full-quantity price.}
\label{fig:public-value-pricing}
\end{figure}
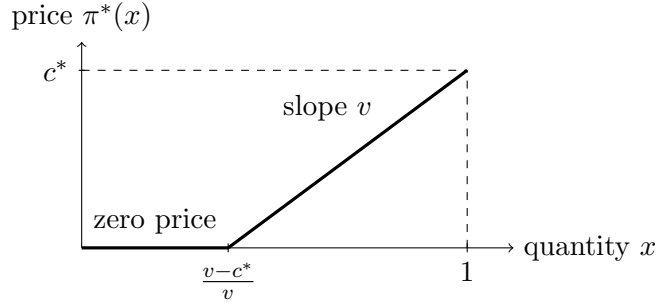

\subsection{Optimal Pricing under Public ROI Constraint}
\label{sec:public-roi}

We now study the setting in which the ROI constraint \(r\) is publicly
known and the valuation \(v\) remains private. Since \(r\) is common
knowledge, the unit-payment cap \(c=v/r\) is determined by \(v\), so a
bidder's type can be identified by \(v\) alone. Let \(F\) denote the
cumulative distribution function of the bidder's valuation on
\([0,\bar v]\) and \(f\) its density. We say \(F\) satisfies
\emph{decreasing marginal revenue} (DMR) if \(1-F(v)-vf(v)\) is strictly
decreasing on \([0,\bar v]\). 

For a convex pricing function \(\pi\) and \(v>0\), at any \(x\in(0,1)\)
where \(\pi\) is differentiable, the bidder purchases at least \(x\) if and only if both the marginal-price condition
\(v\ge\pi'(x)\) and the affordability condition \(c\ge\pi(x)/x\)
hold. Substituting \(c=v/r\), these two conditions are equivalent to
\[
v\ge
\max\left\{\pi'(x),\,r\frac{\pi(x)}{x}\right\}.
\]
We call \(\pi'(x)\) the \emph{marginal-price cutoff} and \(r\pi(x)/x\)
the \emph{affordability cutoff}. These cutoffs capture two distinct
restrictions on demand. If the affordability cutoff is higher, a bidder
at the threshold for purchasing \(x\) exactly meets her ROI requirement
but still has positive marginal utility \(v-\pi'(x)\). If the
marginal-price cutoff is higher, the bidder at this threshold has zero
marginal utility but slack in her ROI constraint. In either case, one
condition determines the purchase threshold while the other is slack.

This motivates a pricing function for which the two cutoffs are equal
at every positive quantity,
\[
\pi'(x)=r\frac{\pi(x)}{x},\qquad x>0.
\]
Both conditions then bind at the purchase threshold, eliminating
slackness.
Solving this differential equation gives the power-law pricing function
\[
\pi_z(x)=\frac{z}{r}x^r,
\]
where \(z\ge0\) is a scale parameter and \(\pi_z(1)=z/r\).
Under DMR, the following theorem shows that the power-law family obtained
by matching the marginal-price and affordability cutoffs contains a
revenue-maximizing pricing function. The optimal scale within this family
is unique.

\begin{theorem}\label{thm:public-roi-optimal-pricing}
Consider a single bidder with publicly known ROI constraint
\(r\in[1,\bar r]\) and private valuation \(v\).
Suppose the valuation distribution \(F\) has a continuous density \(f\)
on \([0,\bar v]\) and satisfies DMR.
An optimal pricing function is
\[
\pi^{*}(x)=\frac{v^{*}}{r}x^{r},
\]
where \(v^{*}\in(0,\bar v)\) is the unique solution of
\[
\int_{0}^{1}
x^{r-1}\left[
1-F\!\left(v^{*}x^{r-1}\right)
-v^{*}x^{r-1}f\!\left(v^{*}x^{r-1}\right)
\right]\,dx
=0.
\]
\end{theorem}

Figure~\ref{fig:public-roi-pricing} illustrates the optimal pricing
function for a fixed \(r>1\). It starts at zero and reaches price
\(v^*/r\) at full quantity, with the theorem's condition determining
\(v^*\) for the given \(r\) and valuation distribution.
When \(r=1\), the pricing function is linear, as in a standard
posted-price auction; when \(r>1\), it is strictly convex, staying
relatively flat for small quantities and growing steeper near full
quantity. For \(r>1\), bidders with \(0<v<v^*\) choose an interior quantity,
while bidders with \(v\geq v^*\) choose the full quantity and pay
\(v^*/r\).

\begin{figure}[t]
\centering
\begin{tikzpicture}[x=5.1cm,y=2.35cm]
  \draw[->] (0,0) -- (1.12,0) node[right] {quantity \(x\)};
  \draw[->] (0,0) -- (0,1.16) node[above] {price \(\pi^*(x)\)};

  \draw[very thick,domain=0:1,samples=60,smooth,variable=\x]
    plot ({\x},{\x*\x});

  \draw[dashed] (1,0) -- (1,1);
  \draw[dashed] (0,1) -- (1,1);
  \draw (1,0.025) -- (1,-0.025) node[below] {\(1\)};
  \draw (0.012,1) -- (-0.012,1) node[left] {\(v^*/r\)};

  \node[below right] at (0.66,0.44) {\(r>1\)};
\end{tikzpicture}
\caption{Optimal pricing with a public ROI constraint \(r>1\).
The price rises from zero according to the power law
\(\pi^*(x)=(v^*/r)x^r\), reaching the full-quantity price \(v^*/r\).}
\Description{A single convex power-law pricing curve for a fixed ROI constraint greater than one, rising from zero to the full-quantity price v-star divided by r.}
\label{fig:public-roi-pricing}
\end{figure}
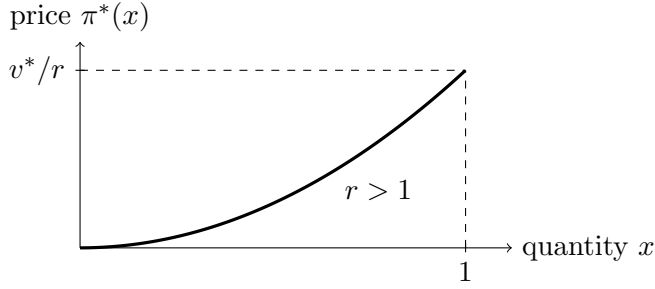

\section{Related Work}
\label{sec:related-work}

Related work spans auto-bidding in repeated advertising auctions and
mechanism design under ROI and/or budget
constraints. In online advertising, ROI constraints are often managed
through auto-bidding, in which an automated agent bids on an
advertiser's behalf subject to a budget or ROI constraint. Most
auto-bidding work studies repeated auctions, with the constraints
imposed jointly across all auctions rather than the one-shot setting
considered here, and mainly asks how bidders should bid or learn, what
outcomes arise, and how efficiently items are allocated. Representative
works develop bidding and learning methods under these constraints
\cite{conf/wine/AggarwalBM19,conf/kdd/HeCWPTYXZ21,conf/www/FengPW23,
conf/ijcai/LiuS23,conf/aistats/GolrezaeiJLM23,conf/www/Aggarwal2025},
or study auction efficiency and how bidders adjust bids to meet spending
limits \cite{conf/www/DengMMZ21,journal/ms/Conitzer2022,journal/or/Conitzer2022,
conf/colt/Lucier2024}. Aggarwal et al.~\cite{autobidding-survey2024}
survey these directions.

Within mechanism design, Golrezaei et al.~\cite{conf/www/GolrezaeiLL21}
study revenue maximization for utility-maximizing bidders under ex ante
ROI constraints, and Ni and Tang~\cite{conf/aaai/NiT22} characterize
incentive compatibility with general ex ante constraints. A separate
line changes the bidder objective from utility maximization to value
maximization, or allows utility and value maximizers to coexist
\cite{conf/www/WilkensCN17,conf/sigecom/BalseiroDMMZ21,conf/sigecom/BalseiroDMMZ22,
conf/ijcai/XingZZYXWC23,conf/aaai/Lv2023,conf/aaai/Bei2025}. These models
provide important foundations for auto-bidding, but differ from our
one-shot problem in bidder objective, information structure, or design
objective.
For ex post ROI constraints, Cavallo et al.~\cite{conf/www/CavalloKSW17}
study deterministic mechanisms with a common public ROI constraint,
while Tang et al.~\cite{Tang2024AAMAS} study welfare approximation
with private valuations and a private binary ROI class.

The work closest to ours is that of
Lv et al.~\cite{conf/wine/Lv2023,journal/AI/Lv2026}, who characterize
truthful auctions for utility-maximizing bidders with
ex post ROI constraints, derive an optimal randomized auction for a
single bidder under DMR, and analyze optimal deterministic auctions
and revenue approximation for multiple bidders. All these results
assume public ROI constraints, so their setting is single-dimensional.
In the fully private setting, where both valuations and ROI constraints
are private, we characterize truthful mechanisms, construct
asymptotically optimal deterministic mechanisms, and bound the revenue
gap between deterministic and randomized truthful mechanisms. These
results generalize their characterization and revenue analysis to the
two-dimensional setting. In particular, we show that applying Myerson's optimal auction
directly to unit-payment caps in the fully private setting need not be
pointwise DSIC. Under regularity, its revenue is a tight upper bound
for deterministic truthful mechanisms, rather than an optimum
attained by a truthful mechanism.
For a single bidder with a public
ROI constraint, our optimal pricing function under DMR coincides with
their optimal auction after translating their value variable into our
unit-payment cap. We provide a substantially different derivation with
clear economic intuition for why the optimal price follows a power law.
In our convex pricing formulation, matching the marginal-price cutoff
with the ROI affordability cutoff yields this pricing rule,
complementing their allocation-based characterization.
We also prove that convex pricing functions suffice for revenue
maximization in the single-bidder setting when both the valuation and
the ROI constraint are private, and derive optimal pricing in the
public-valuation setting.

Earlier work on auction design with budget constraints is also related
\cite{journal/RES/CheGale98,CHE2000,
journal/EL/LaffontRobert96,journal/RES/BenoitKrishna01,
journal/JET/PaiVohra14,conf/sigecom/DevanurW17}.
Particularly relevant is Che and Gale~\cite{CHE2000}, who derive the
optimal pricing function for a single buyer with
a private valuation and budget. Our pricing-function results extend this
line of work to ROI constraints.
A budget constraint $b$ requires $p\leq b$, independently of the allocation $x$,
whereas an ROI constraint requires $p\leq (v/r)x$, so the maximum
permitted payment scales with the allocation. This allocation
dependence makes incentive analysis and revenue maximization more
demanding under ROI constraints.
More broadly, our problem connects to multidimensional mechanism design
\cite{journal/JET/MCAFEE1988,journal/Econometrica/Krishna01,
journal/JME/ROCHET1985}, where our mechanism-restriction technique,
used to reduce dimensionality in deriving the payment identity,
may be a useful tool.

\section{Conclusion}

We studied truthful auctions for bidders with private valuations and
ROI constraints. Representing each bidder by her valuation and
unit-payment cap gives a characterization of truthful mechanisms in
which the allocation function uniquely determines the payment function.
For multiple bidders, our \(\sigma\)-increment mechanisms approach the
best expected revenue among deterministic truthful mechanisms under
regularity. Their limiting revenue is also at least a \(1/\bar r\)
fraction of the best expected revenue among all truthful mechanisms,
including randomized ones.
For a single bidder, every truthful mechanism can be replaced by a
convex pricing function without reducing any type's payment. With a
public valuation, the optimal pricing function is zero up to a quantity
threshold and then increases linearly. With a public ROI constraint, the
optimal pricing function follows a power law.

\appendix
\section{Proofs for Section~\ref{sec:truthful-mechanisms} (Truthful Mechanisms)}
\label{app:truthful-mechanisms}

\subsection*{Proof of Lemma~\ref{lem:unit-payment-monotonicity}}

Let $t=(v,c)$ and $t'=(v',c')$ be any two types in $T$.

For the first claim, suppose that $q(t)<q(t')$. Since $q(t')\le c'$ by IR,
reporting $t$ is safe for type $t'$. If $x(t)\ge x(t')$, then
\[
u(t;t')=(v'-q(t))x(t)>(v'-q(t'))x(t')=u(t';t'),
\]
contradicting DSIC. Hence $x(t)<x(t')$.

For the second claim, first suppose that $x(t)<x(t')$. If
$c<q(t')$, then $q(t)\leq c<q(t')$ by IR. Otherwise,
reporting $t'$ is safe for type $t$. If $q(t)>q(t')$, then
\[
u(t;t)=(v-q(t))x(t)<(v-q(t'))x(t')=u(t';t),
\]
contradicting DSIC. Hence $q(t)\leq q(t')$.

It remains to consider $x(t)=x(t')$. If the two unit payments were
unequal, the first claim, applied in the appropriate order, would force
a strict inequality between the allocations. Hence $q(t)=q(t')$,
completing the proof.
\qed

\subsection*{Proof of Theorem~\ref{thm:Allocation}}

\noindent
\textbf{(A1)}: By DSIC, a bidder with type $t'$ prefers reporting
truthfully over $t$. Since $q(t) \leq c'$, reporting $t$ is safe, which
yields
\[
v'x(t')-p(t')\geq v'x(t)-p(t).
\]
Rearranging gives
\begin{equation}
p(t)-p(t')\geq v'\bigl(x(t)-x(t')\bigr).
\label{eq:ic_t_prime}
\end{equation}
For a bidder with type $t$: if reporting $t'$ is unsafe (i.e.,
$q(t') > c$), then since $q(t) \leq c$ by IR, we have $q(t) < q(t')$,
and Lemma~\ref{lem:unit-payment-monotonicity} gives $x(t) < x(t')$. If
instead reporting $t'$ is safe, DSIC requires $u(t;t) \geq u(t';t)$,
which gives
\[
vx(t)-p(t)\geq vx(t')-p(t').
\]
Rearranging gives
\begin{equation}
v\bigl(x(t)-x(t')\bigr)\geq p(t)-p(t').
\label{eq:ic_t}
\end{equation}
Combining \eqref{eq:ic_t_prime} and \eqref{eq:ic_t} gives
$v(x(t)-x(t'))\geq v'(x(t)-x(t'))$, so
$(v - v')(x(t) - x(t')) \geq 0$. Since $v < v'$, this forces
$x(t) \leq x(t')$.

\noindent
\textbf{(A2)}: IR implies $q(t') \leq c'$, and $c' < q(t)$ then gives
$q(t')\leq c'<q(t)$, so $q(t')<q(t)$.
Lemma~\ref{lem:unit-payment-monotonicity} then gives $x(t')<x(t)$.

\noindent
\textbf{(A3)}: Since IR guarantees \(q(v,c) \leq c\), it follows from
$c\le c'$ that \(q(v,c) \leq c'\); \textbf{(A1)} then gives
\(x(v,c) \leq x(v',c')\).
\qed

\subsection*{Proof of Lemma~\ref{lem:allocation_A_E}}

Let the IR mechanism $(x,p)$ satisfy \textbf{(A1)}. In an IR mechanism,
\textbf{(A3)} follows from \textbf{(A1)}, so $x(\cdot,c)$ and
$x(\cdot,c')$ are nondecreasing and hence continuous almost everywhere.

Take any $c<c'$ and an interior value $v$ at which $x(\cdot,c')$ is
continuous. For every sufficiently small $\delta>0$, \textbf{(A3)} gives
\[
x(v,c)\le x(v+\delta,c').
\]
Letting $\delta\downarrow0$ proves \textbf{(A4)}.

Now suppose in addition that $v$ is a continuity point of $x(\cdot,c)$
and that $q(v,c')\le c$. For every sufficiently small $\delta>0$,
\textbf{(A1)}, applied to $(v,c')$ and $(v+\delta,c)$, gives
\[
x(v,c')\le x(v+\delta,c).
\]
Letting $\delta\downarrow0$ gives the reverse inequality
$x(v,c')\le x(v,c)$, which together with \textbf{(A4)} proves
\textbf{(A5)} almost everywhere.
\qed

\subsection*{Proof of Theorem~\ref{thm:Payment}}

Since $(x,p)$ is truthful on $T$, it remains truthful when restricted to
either of the one-dimensional subsets $T^=$ and $T^c$ of the type space.
We first determine $p(c,c)$ using types in $T^=$, and then determine
$p(v,c)-p(c,c)$ using types in $T^c$, so that combining the two
expressions gives $p(v,c)$.

On $T^=$, consider a bidder with type $(s,s)$ who reports $(s',s')$.
If $p(s',s')\leq s x(s',s')$, the report satisfies her ROI constraint,
so DSIC implies $s x(s,s)-p(s,s)\geq s x(s',s')-p(s',s')$.
Otherwise, $s x(s',s')-p(s',s')<0$, whereas IR ensures that she
obtains utility $s x(s,s)-p(s,s)\geq0$ by reporting truthfully, so the
same inequality still holds. Interchanging $s$ and $s'$, we therefore
have, for every $s,s'\in[0,\bar v]$,
\[
\begin{aligned}
s x(s,s)-p(s,s)
&\geq s x(s',s')-p(s',s'),\\
s' x(s',s')-p(s',s')
&\geq s' x(s,s)-p(s,s).
\end{aligned}
\]
Treating $x(s,s)$ and $p(s,s)$ as functions of the single variable $s$,
these inequalities have the same form as those used to derive Myerson's
payment identity. Following the same proof steps and using $p(0,0)=0$,
we obtain
\begin{equation}\label{eq:diagonal}
    p(c,c)=c x(c,c)-\int_0^c x(z,z)\,dz.
\end{equation}

On $T^c$, the unit-payment cap is fixed. Since IR gives
$p(v',c)\leq c x(v',c)$ for every $(v',c)\in T^c$, deviations within
this subset never violate the ROI constraint, so the bidder behaves as
a standard utility maximizer with private value $v$. Hence
\[
\begin{aligned}
vx(v,c)-p(v,c)
&\ge vx(v',c)-p(v',c),\\
v'x(v',c)-p(v',c)
&\ge v'x(v,c)-p(v,c).
\end{aligned}
\]
Treating $x(v,c)$ and $p(v,c)$ as functions of the single variable $v$
with $c$ fixed, we again follow the proof steps used to derive Myerson's
payment identity, now starting at value $c$, to obtain
\begin{equation}\label{eq:horizontal}
    p(v,c)-p(c,c)
    =v x(v,c)-c x(c,c)-\int_c^v x(z,c)\,dz.
\end{equation}
Substituting \eqref{eq:diagonal} into \eqref{eq:horizontal} yields
\[
    p(v,c)
    =v x(v,c)-\int_0^c x(z,z)\,dz-\int_c^v x(z,c)\,dz,
\]
as claimed.
\qed

\subsection*{Proof of Theorem~\ref{thm:characterisation}}
\label{app:proof-characterisation}

We first establish two auxiliary lemmas.

\begin{lemma}\label{lem:unit-payment-monotone}
If a mechanism $(x,p)$ satisfies \emph{\textbf{(A1)}},
\emph{\textbf{(P)}}, and \emph{\textbf{(IR)}}, then
\begin{enumerate}
    \item if $v<v'$, then $q(v,c)\leq q(v',c)$; and
    \item if $c<c'$, then $q(c,c)\leq q(c',c')$.
\end{enumerate}
\end{lemma}

For the first claim, let $t=(v,c)$ and $t'=(v',c)$ with $v<v'$.
By \textbf{(A1)} and IR, $x(t')\geq x(t)$. If $x(t)=0$, then
$q(t)=0\leq q(t')$, so assume $x(t')\geq x(t)>0$.
Using \textbf{(P)} and $x(z,c)\leq x(t')$ for $z\in[v,v']$, we have
\[
\begin{aligned}
p(t')-p(t)
&=v\bigl(x(t')-x(t)\bigr)+\int_v^{v'}\bigl(x(t')-x(z,c)\bigr)\,dz\\
&\geq v\bigl(x(t')-x(t)\bigr).
\end{aligned}
\]
Since $q(t)\leq c\leq v$ by IR, it follows that
\[
q(t')-q(t)\geq\frac{(v-q(t))\bigl(x(t')-x(t)\bigr)}{x(t')}\geq0.
\]

For the second claim, let $t=(c,c)$ and $t'=(c',c')$ with $c<c'$.
The zero-allocation case is handled as above. Otherwise,
\textbf{(A1)}, IR, and \textbf{(P)} similarly give
\[
\begin{aligned}
p(t')-p(t)
&=c\bigl(x(t')-x(t)\bigr)+\int_c^{c'}\bigl(x(t')-x(z,z)\bigr)\,dz\\
&\geq c\bigl(x(t')-x(t)\bigr).
\end{aligned}
\]
The same calculation, with $c$ in place of $v$, yields $q(t)\leq q(t')$.
\qed

\begin{lemma}\label{lem:diagonal-allocation}
If a mechanism $(x,p)$ satisfies \emph{\textbf{(A1)}} and
\emph{\textbf{(IR)}}, then, for any fixed cap $c$ and almost every $v$
with $(v,c)\in T$:
\begin{description}
    \item[(A6)] $x(v,c)\leq x(v,v)$.
    \item[(A7)] If $q(v,v)\leq c$, then $x(v,c)=x(v,v)$.
\end{description}
\end{lemma}

\begin{proof}
\textbf{(IR)} and \textbf{(A1)} imply \textbf{(A3)}, which makes both
$x(v,v)$ and $x(v,c)$ nondecreasing in $v$, with $c$ fixed, and hence
continuous almost everywhere. Thus almost every $v$ with $(v,c)\in T$
is an interior point where both functions are continuous. At any such $v$,
\textbf{(A3)} gives
$x(v,c)\leq x(v+\delta,v+\delta)$ for sufficiently small $\delta>0$;
letting $\delta\downarrow0$ yields
\[
x(v,c)\le x(v,v),
\]
which proves \textbf{(A6)}. If $q(v,v)\le c$, applying \textbf{(A1)} to
$(v,v)$ and $(v+\delta,c)$ and letting $\delta\downarrow0$ gives the
reverse inequality, proving \textbf{(A7)}.
Both claims therefore hold for almost every $v$ with $(v,c)\in T$.
\end{proof}

We now prove Theorem~\ref{thm:characterisation}.

The necessity of \textbf{(A1)} and \textbf{(P)} follows from
Theorems~\ref{thm:Allocation} and~\ref{thm:Payment} respectively, and
the necessity of \textbf{(IR)} is immediate.

For sufficiency, suppose $(x,p)$ satisfies \textbf{(A1)}, \textbf{(P)},
and \textbf{(IR)}. Fix types $t=(v,c)$ and $t'=(v',c')$.
The measure-zero exceptions in Lemmas~\ref{lem:allocation_A_E}
and~\ref{lem:diagonal-allocation} do not affect the integrals below.
If reporting $t'$
violates the ROI constraint of type $t$, then $u(t';t)=-\infty$ and
there is nothing to prove, so assume $q(t')\le c$. By \textbf{(P)}, it
suffices to show
\[
\begin{aligned}
&\Delta:=u(t;t)-u(t';t)
=\int_0^c x(z,z)\,dz+\int_c^v x(z,c)\,dz\\
&\quad-(v-v')x(t')-\int_0^{c'}x(z,z)\,dz-\int_{c'}^{v'}x(z,c')\,dz\ge0.
\end{aligned}
\]
We consider the possible orderings of $c,c',v$, and $v'$.
In each case, we first cancel the common part of the two integrals
starting at zero, then split the remaining integrals at the ordered
endpoints. We also write $(v'-v)x(t')$ as the integral of the constant
$x(t')$ over $[v,v']$ when $v\le v'$, or as its negative over
$[v',v]$ when $v>v'$.

\noindent
\textbf{Case 1: $c\le v\le c'\le v'$.}
Canceling the diagonal integrals over $[0,c]$ and then splitting at $v$
gives
\[
\begin{aligned}
\Delta
&=(v'-v)x(t')-\int_c^v x(z,z)\,dz-\int_v^{c'}x(z,z)\,dz\\
&\quad+\int_c^v x(z,c)\,dz-\int_{c'}^{v'}x(z,c')\,dz\\
&=\int_v^{c'}\bigl(x(t')-x(z,z)\bigr)\,dz
+\int_{c'}^{v'}\bigl(x(t')-x(z,c')\bigr)\,dz\\
&\quad+\int_c^v\bigl(x(z,c)-x(z,z)\bigr)\,dz.
\end{aligned}
\]
The first two integrands are nonnegative by \textbf{(A3)}. Moreover,
\textbf{(A3)} and Lemma~\ref{lem:unit-payment-monotone} imply
$q(z,z)\le q(t')\le c$ for $z\in[c,v]$. Property \textbf{(A7)} therefore gives
$x(z,c)=x(z,z)$ for almost every $z\in[c,v]$.
Hence $\Delta\ge0$.

\noindent
\textbf{Case 2: $c\le c'\le v\le v'$.}
Canceling the diagonal integrals over $[0,c]$, then splitting the
integral of $x(z,c)$ at $c'$ and that of $x(z,c')$ at $v$, gives
\[
\begin{aligned}
\Delta
&=(v'-v)x(t')-\int_c^{c'}x(z,z)\,dz
+\int_c^{c'}x(z,c)\,dz\\
&\quad+\int_{c'}^v x(z,c)\,dz-\int_{c'}^v x(z,c')\,dz
-\int_v^{v'}x(z,c')\,dz\\
&=(v'-v)x(t')+\int_c^{c'}\bigl(x(z,c)-x(z,z)\bigr)\,dz\\
&\quad+\int_{c'}^v\bigl(x(z,c)-x(z,c')\bigr)\,dz-\int_v^{v'}x(z,c')\,dz.
\end{aligned}
\]
Because $q(t')\le c$, \textbf{(A3)} and
Lemma~\ref{lem:unit-payment-monotone} imply the relevant affordability
conditions. Property \textbf{(A7)} makes the first difference zero almost
everywhere, while \textbf{(A5)}
for the fixed pair $(c,c')$ makes the second difference zero almost
everywhere. Consequently,
\[
\Delta=\int_v^{v'}\bigl(x(t')-x(z,c')\bigr)\,dz\ge0
\]
by \textbf{(A3)}.

\noindent
\textbf{Case 3: $c\le c'\le v'<v$.}
Canceling the diagonal integrals over $[0,c]$, then splitting the
integral of $x(z,c)$ at $c'$ and $v'$, gives
\[
\begin{aligned}
\Delta
&=-(v-v')x(t')-\int_c^{c'}x(z,z)\,dz
+\int_c^{c'}x(z,c)\,dz\\
&\quad+\int_{c'}^{v'}x(z,c)\,dz+\int_{v'}^v x(z,c)\,dz
-\int_{c'}^{v'}x(z,c')\,dz\\
&=\int_c^{c'}\bigl(x(z,c)-x(z,z)\bigr)\,dz\\
&\quad+\int_{c'}^{v'}\bigl(x(z,c)-x(z,c')\bigr)\,dz
+\int_{v'}^v\bigl(x(z,c)-x(t')\bigr)\,dz.
\end{aligned}
\]
As in Case 2, \textbf{(A7)} and \textbf{(A5)} make the first two
integrals zero. Hence
\[
\Delta=\int_{v'}^v\bigl(x(z,c)-x(t')\bigr)\,dz.
\]
For $z\in(v',v]$, the assumption $q(t')\le c$ and
\textbf{(A1)} imply $x(z,c)\ge x(t')$. The omitted endpoint has measure
zero, so $\Delta\ge0$.

\noindent
\textbf{Case 4: $c'\le v'\le c\le v$.}
Canceling the diagonal integrals over $[0,c']$ and then splitting at
$v'$ gives
\[
\begin{aligned}
\Delta
&=-(v-v')x(t')+\int_{c'}^{v'}x(z,z)\,dz
+\int_{v'}^c x(z,z)\,dz\\
&\quad+\int_c^v x(z,c)\,dz-\int_{c'}^{v'}x(z,c')\,dz\\
&=\int_{c'}^{v'}\bigl(x(z,z)-x(z,c')\bigr)\,dz
+\int_{v'}^c\bigl(x(z,z)-x(t')\bigr)\,dz\\
&\quad+\int_c^v\bigl(x(z,c)-x(t')\bigr)\,dz.
\end{aligned}
\]
The second and third integrands are nonnegative almost everywhere by
\textbf{(A3)}; only endpoints at which the two valuations coincide can be
excluded. The first is nonnegative almost everywhere by \textbf{(A6)}.
Hence $\Delta\ge0$.

\noindent
\textbf{Case 5: $c'\le c\le v'\le v$.}
Canceling the diagonal integrals over $[0,c']$, then splitting the
integral of $x(z,c)$ at $v'$ and that of $x(z,c')$ at $c$, gives
\[
\begin{aligned}
\Delta
&=-(v-v')x(t')+\int_{c'}^c x(z,z)\,dz
+\int_c^{v'}x(z,c)\,dz\\
&\quad+\int_{v'}^v x(z,c)\,dz-\int_{c'}^c x(z,c')\,dz
-\int_c^{v'}x(z,c')\,dz\\
&=\int_{c'}^c\bigl(x(z,z)-x(z,c')\bigr)\,dz
+\int_c^{v'}\bigl(x(z,c)-x(z,c')\bigr)\,dz\\
&\quad+\int_{v'}^v\bigl(x(z,c)-x(t')\bigr)\,dz.
\end{aligned}
\]
The first integrand is nonnegative almost everywhere by \textbf{(A6)};
the second is nonnegative
almost everywhere by \textbf{(A4)} for the fixed pair $(c',c)$; and the
third is nonnegative almost everywhere by \textbf{(A1)}, with only the
endpoint $z=v'$ excluded. Hence $\Delta\ge0$.

\noindent
\textbf{Case 6: $c'\le c\le v<v'$.}
Canceling the diagonal integrals over $[0,c']$, then splitting the
integral of $x(z,c')$ at $c$ and $v$, gives
\[
\begin{aligned}
\Delta
&=(v'-v)x(t')+\int_{c'}^c x(z,z)\,dz
+\int_c^v x(z,c)\,dz\\
&\quad-\int_{c'}^c x(z,c')\,dz-\int_c^v x(z,c')\,dz
-\int_v^{v'}x(z,c')\,dz\\
&=\int_{c'}^c\bigl(x(z,z)-x(z,c')\bigr)\,dz
+\int_c^v\bigl(x(z,c)-x(z,c')\bigr)\,dz\\
&\quad+\int_v^{v'}\bigl(x(t')-x(z,c')\bigr)\,dz.
\end{aligned}
\]
The first integrand is nonnegative almost everywhere by \textbf{(A6)},
the second is nonnegative almost everywhere by \textbf{(A4)} for the
fixed pair $(c',c)$, and the last is nonnegative pointwise by
\textbf{(A3)}. Thus $\Delta\ge0$ in every case, establishing DSIC;
condition \textbf{(IR)} establishes individual rationality.
\qed

\subsection*{Proof of Lemma~\ref{lem:diagonal-determination}}

Fix a type $t=(v,c)\in T\setminus T^{=}$ with $v<\bar v$.
Property \textbf{(A3)} implies that the diagonal allocation is
nondecreasing and gives $x(c,c)\le x(v,c)$.
For every sufficiently small $\delta>0$, the same property gives
$x(v,c)\le x(v+\delta,v+\delta)$. Diagonal continuity therefore yields
\[
x(c,c)\le x(v,c)\le x(v,v).
\]
By continuity, there is a $z^*\in[c,v]$ such that
$x(z^*,z^*)=x(v,c)$. Lemma~\ref{lem:unit-payment-monotonicity} and IR give
\[
q(z^*,z^*)=q(v,c)\le c.
\]
Thus $z^*$ is feasible in the supremum in
\eqref{eq:diagonal-determination}, so that supremum is at least $x(v,c)$.

Conversely, let $z\in[c,v]$ satisfy $q(z,z)\le c$.
If $z<v$, property \textbf{(A1)} directly gives $x(z,z)\le x(v,c)$.
If $z=v$, choose diagonal types $(z_k,z_k)$ with $c\le z_k<v$ and
$z_k\uparrow v$. Diagonal monotonicity and
Lemma~\ref{lem:unit-payment-monotonicity} imply
$q(z_k,z_k)\le q(v,v)\le c$. Applying \textbf{(A1)} gives
$x(z_k,z_k)\le x(v,c)$; taking the limit and using diagonal continuity
then gives $x(v,v)\le x(v,c)$.
Every allocation in the supremum is therefore at most $x(v,c)$,
which proves the identity.
\qed

\section{Proofs for Section~\ref{sec:deterministic-mechanisms} (Deterministic Mechanisms)}
\label{app:deterministic-mechanisms}

\subsection*{Proof of Theorem~\ref{thm:characterisation-deterministic}}

Suppose first that $(x,p)$ is deterministic and truthful.

For \textbf{(DA1)}, let $c<c'$ and suppose $x(v',c')=0$. If $x(v,c)=1$,
then IR gives $p(v,c)\le c<c'$, so type $(v',c')$ can safely report
$(v,c)$. The deviation yields utility $v'-p(v,c)>0$, contradicting
DSIC. Hence $x(v,c)=0$.

For \textbf{(DA2)}, take $v,v'>c$ and suppose without loss of generality
that $v<v'$. Property \textbf{(A3)} gives $x(v,c)\le x(v',c)$. If the
inequality were strict, determinism would give $x(v,c)=0$ and
$x(v',c)=1$; type $(v,c)$ could then safely report $(v',c)$, and IR
would give $p(v',c)\le c<v$, making the deviation strictly profitable.
Thus $x(v,c)=x(v',c)$.

To prove \textbf{(DP)}, first suppose $x(v,c)=0$. By
\textbf{(DA1)}, $x(z,z)=0$ for every $z<c$, and hence both sides of
\textbf{(DP)} are zero. If $x(v,c)=1$, then \textbf{(P)} and
\textbf{(DA2)} give
\[
\begin{aligned}
p(v,c)
&=v-\int_0^c x(z,z)\,dz-\int_c^v x(z,c)\,dz\\
&=v-\int_0^c x(z,z)\,dz-(v-c)\\
&=c-\int_0^c x(z,z)\,dz,
\end{aligned}
\]
which is \textbf{(DP)}.

It remains to verify \textbf{(DA3)}. Suppose some off-diagonal type
wins, and let $\kappa$ be the infimum of all caps at which such a type
wins. Conditions \textbf{(DA1)} and \textbf{(DA2)} imply that every
off-diagonal type with cap above $\kappa$ wins. Suppose, toward a
contradiction, that the off-diagonal types with cap $\kappa$ lose.
Property \textbf{(DA1)} then implies $x(z,z)=0$ for $z<\kappa$ and
$x(z,z)=1$ for $z>\kappa$. By \textbf{(DP)}, every winning report with
cap above $\kappa$ pays $\kappa$, so a type $(v,\kappa)$ with $v>\kappa$
could report such a winning type, safely pay $\kappa$, and obtain
positive utility, contradicting DSIC. Hence every off-diagonal type
with cap $\kappa$ wins.

Conversely, suppose a deterministic mechanism satisfies
\textbf{(DA1)}--\textbf{(DA3)} and \textbf{(DP)}. Conditions
\textbf{(DA1)}--\textbf{(DA3)} imply a threshold $\kappa\in[0,\bar v]$
such that the off-diagonal allocation is zero below $\kappa$ and one at
and above $\kappa$. If no off-diagonal type wins, take $\kappa=\bar v$;
\textbf{(DA1)} then rules out every diagonal winner except possibly
$(\bar v,\bar v)$. The diagonal allocation is likewise zero below
$\kappa$ and one above $\kappa$, while $x(\kappa,\kappa)$ may be either
zero or one. Consequently,
\[
\int_0^c x(z,z)\,dz=\max\{c-\kappa,0\},
\]
and \textbf{(DP)} gives $p(v,c)=\kappa$ for every winner and zero for every
loser, with $p(\kappa,\kappa)=\kappa x(\kappa,\kappa)$.

By Theorem~\ref{thm:characterisation}, to establish truthfulness it
suffices to verify \textbf{(A1)}, \textbf{(P)}, and \textbf{(IR)}.
Condition \textbf{(DP)} implies \textbf{(IR)}.
For \textbf{(A1)}, consider $t=(v,c)$ and
$t'=(v',c')$ with $v<v'$ and $q(t)\le c'$. The claim is immediate if $t$
loses. If $t$ wins, then $q(t)=\kappa$, so $c'\ge \kappa$.
If $c'>\kappa$, the threshold structure gives $x(t')=1$.
If $c'=\kappa$, then $v'>v\ge \kappa$, so \textbf{(DA3)} again gives
$x(t')=1$.
Thus $x(t)\le x(t')$.

Finally, to verify \textbf{(P)}, note that \textbf{(DA2)} gives
$x(z,c)=x(v,c)$ for $c<z<v$. Hence
\[
\begin{aligned}
&v x(v,c)-\int_0^c x(z,z)\,dz-\int_c^v x(z,c)\,dz\\
&\quad=v x(v,c)-\int_0^c x(z,z)\,dz-(v-c)x(v,c)\\
&\quad=c x(v,c)-\int_0^c x(z,z)\,dz\\
&\quad=p(v,c),
\end{aligned}
\]
where the last equality is \textbf{(DP)}. The calculation also holds
when $v=c$, since the integral from $c$ to $v$ is zero.
Thus \textbf{(P)} holds.
\qed

\subsection*{Proof of Theorem~\ref{thm:optimal-deterministic}}

Fix $\sigma>0$. Recall that $M^0=(x^0,p^0)$ is the Myerson-cap
mechanism, and $x_i^\sigma(\mathbf t)$ is the allocation assigned to
bidder $i$ at type profile $\mathbf t$ under $M^\sigma$.

We first verify that the allocation rule assigns the
item to at most one bidder. If $x_i^\sigma(\mathbf t)=1$, then
\[
c_i\geq\kappa_i(\mathbf t_{-i})+\sigma
>\kappa_i(\mathbf t_{-i}).
\]
Since $x_i^0$ is nondecreasing in $c_i$ and $\kappa_i$ is the infimum
of its winning caps, every cap strictly above $\kappa_i$ wins under
$M^0$. Thus $x_i^\sigma(\mathbf t)\leq x_i^0(\mathbf t)$ for every
bidder $i$. By construction, $M^0$ selects at most one winner, so
\[
\sum_i x_i^\sigma(\mathbf t)
\leq\sum_i x_i^0(\mathbf t)\leq1.
\]

For bidder $i$, conditional on the other reports, $M^\sigma$ has the
winning threshold
\[
\tau_i(\mathbf t_{-i})
:=\kappa_i(\mathbf t_{-i})+\sigma.
\]
If $\kappa_i(\mathbf t_{-i})=\infty$, every report by bidder $i$ gives
her zero allocation and zero payment, so no misreport can improve her
utility. We may therefore suppose the threshold is finite.
Its allocation depends only on the unit-payment cap and equals zero below
$\tau_i$ and one at and above $\tau_i$. Moreover,
\[
\int_0^{c_i}x_i^\sigma((z,z),\mathbf t_{-i})\,dz
=\max\{c_i-\tau_i,0\}.
\]
Substituting this integral identity and the threshold allocation into
the right-hand side of \textbf{(DP)} yields zero for losers and $\tau_i$
for winners, exactly the payments prescribed by $M^\sigma$.
Thus \textbf{(DP)} holds.
Conditions
\textbf{(DA1)}--\textbf{(DA3)} follow because the allocation depends only
on the cap and winning includes equality at the threshold.
Theorem~\ref{thm:characterisation-deterministic} therefore
implies that $M^\sigma$ is truthful.

It remains to prove revenue convergence. By independence and atomlessness
of the bidders' cap distributions,
\[
\Pr\!\left[
c_i=\kappa_i(\mathbf t_{-i})
\text{ for some }i
\right]=0.
\]
Consider a cap profile outside this null event. If $M^0$ has no winner,
then $M^\sigma$ has no winner for any $\sigma>0$. If bidder $i$ wins in
$M^0$, then
\[
c_i-\kappa_i(\mathbf t_{-i})>0.
\]
For every sufficiently small $\sigma$, bidder $i$ also wins in
$M^\sigma$ and pays $\kappa_i(\mathbf t_{-i})+\sigma$, which converges
to her payment in $M^0$. So the realized revenue of $M^\sigma$
converges almost surely to that of $M^0$. Since realized revenue is
bounded by $\bar v$, it follows from the dominated convergence theorem
that
\[
\lim_{\sigma\downarrow0}\operatorname{Rev}(M^\sigma)
=\operatorname{Rev}(M^0).
\]
\qed

\subsection*{Proof of Lemma~\ref{lem:fixed-roi-payment-bound}}

For a bidder with valuation $v$ and ROI constraint $r$, the
unit-payment cap is $c=v/r$, so increasing $v$ while holding $r$
constant increases both the valuation and the unit-payment cap. By
(A3), $x_r(\cdot)$ is therefore nondecreasing.

Define $U_r(v):=vx_r(v)-p_r(v)$. Fix $v>0$ and a partition
\[
0=v_0<v_1<\cdots<v_k=v.
\]
For each $j$, IR makes the outcome at $v_{j-1}$ affordable to type $v_j$,
since
\[
p_r(v_{j-1})
\le \frac{v_{j-1}}{r}x_r(v_{j-1})
\le \frac{v_j}{r}x_r(v_{j-1}).
\]
DSIC implies
\[
U_r(v_j)\ge v_j x_r(v_{j-1})-p_r(v_{j-1})=U_r(v_{j-1})+(v_j-v_{j-1})x_r(v_{j-1}).
\]
Rearranging gives
\[
U_r(v_j)-U_r(v_{j-1})
\ge (v_j-v_{j-1})x_r(v_{j-1}).
\]
Summing this inequality over $j=1,\ldots,k$, the intermediate utility terms cancel.
Since $U_r(0)=0$, we obtain
\[
U_r(v)\ge\sum_{j=1}^k(v_j-v_{j-1})x_r(v_{j-1}).
\]
Since $x_r$ is bounded and nondecreasing, letting the mesh of the
partition tend to zero gives
\[
U_r(v)\ge \int_0^v x_r(z)\,dz.
\]
Substituting $U_r(v)=vx_r(v)-p_r(v)$ and rearranging proves the lemma.
\qed

\subsection*{Proof of Lemma~\ref{lem:opt-upper-value}}

Consider an arbitrary truthful mechanism $M=(x,p)$, and fix an ROI
profile
\[
\mathbf r=(r_1,\ldots,r_n).
\]
For every valuation profile $\mathbf v=(v_1,\ldots,v_n)$, define
\[
x_i^{\mathbf r}(\mathbf v)
:=
x_i\bigl((v_1,r_1),\ldots,(v_n,r_n)\bigr).
\]
For fixed $\mathbf r$ and $\mathbf v_{-i}$, increasing $v_i$ also
increases the unit-payment cap $v_i/r_i$, so (A3) implies that
$x_i^{\mathbf r}(v_i,\mathbf v_{-i})$ is nondecreasing in $v_i$.

Define
\begin{equation}
\widetilde p_i^{\mathbf r}(\mathbf v)
:=
v_i x_i^{\mathbf r}(\mathbf v)
-
\int_0^{v_i}x_i^{\mathbf r}(z,\mathbf v_{-i})\,dz.
\label{eq:canonical-payment}
\end{equation}
Since $x_i^{\mathbf r}(\cdot,\mathbf v_{-i})$ is nondecreasing, the pair
$(x^{\mathbf r},\widetilde p^{\mathbf r})$ is a truthful mechanism in the ordinary
single-dimensional quasilinear environment with values $v_i$.

Applying Lemma~\ref{lem:fixed-roi-payment-bound}, while fixing
$\mathbf r$ and $\mathbf v_{-i}$, gives pointwise
\begin{equation}
p_i\bigl((v_1,r_1),\ldots,(v_n,r_n)\bigr)
\le
\widetilde p_i^{\mathbf r}(\mathbf v).
\label{eq:pointwise-payment-domination}
\end{equation}
Therefore,
\begin{equation}
\sum_{i=1}^n
p_i\bigl((v_1,r_1),\ldots,(v_n,r_n)\bigr)
\le
\sum_{i=1}^n\widetilde p_i^{\mathbf r}(\mathbf v).
\end{equation}

By independence, conditioning on the ROI profile being $\mathbf r$ does
not change the distribution of the valuation profile $\mathbf v$, which remains
$F_1\times\cdots\times F_n$. Hence
$(x^{\mathbf r},\widetilde p^{\mathbf r})$ is a feasible ordinary
truthful auction facing exactly the valuation distributions used to
define $\mathrm{OPT}_{r=1}$. By definition of $\mathrm{OPT}_{r=1}$,
\[
\mathbb{E}_{\mathbf v}\!\left[
\sum_{i=1}^n \widetilde p_i^{\mathbf r}(\mathbf v)
\right]
\le
\mathrm{OPT}_{r=1}.
\]
Together with \eqref{eq:pointwise-payment-domination}, this implies
\[
\mathbb{E}_{\mathbf v}\!\left[
\sum_{i=1}^n
p_i\bigl((v_1,r_1),\ldots,(v_n,r_n)\bigr)
\right]
\le
\mathrm{OPT}_{r=1}
\]
for every $\mathbf r$. Taking expectation over the ROI profile gives
\[
\operatorname{Rev}(M)\le \mathrm{OPT}_{r=1}.
\]
Since $M$ was arbitrary,
\[
\mathrm{OPT}\le \mathrm{OPT}_{r=1}.
\]
\qed

\subsection*{Proof of Lemma~\ref{lem:det-lower-value}}

Since $1\le r_i\le\bar r$, we have, pointwise,
\begin{equation}
c_i=\frac{v_i}{r_i}
\ge
\frac{v_i}{\bar r}.
\label{eq:c-dominates-scaled-v}
\end{equation}
So the distribution of the unit-payment cap $c_i$ first-order
stochastically dominates the distribution of $v_i/\bar r$ for every
bidder $i$.

Let $\mathrm{OPT}_c$ denote the optimal revenue in the ordinary
single-dimensional auction with values distributed as
$c_1,\ldots,c_n$. Under regularity, this is the revenue of the
Myerson-cap mechanism $M^0$. The revenue upper bound established above,
together with Theorem~\ref{thm:optimal-deterministic}, shows that the same
value equals $\mathrm{OPT}_{\mathrm{det}}$ in the ROI-constrained
environment. Hence
\begin{equation}
\mathrm{OPT}_{\mathrm{det}}=\mathrm{OPT}_c.
\label{eq:det-c-revenue}
\end{equation}
By revenue monotonicity under first-order stochastic dominance, a
folklore consequence of Myerson's optimal-auction characterization,
\begin{equation}
\mathrm{OPT}_c
\ge
\mathrm{OPT}_{v/\bar r},
\end{equation}
where $\mathrm{OPT}_{v/\bar r}$ denotes the optimal single-item
revenue when bidder $i$ has value $v_i/\bar r$.

Scaling all bidders' values by the common factor $1/\bar r$ scales
every feasible payment, and hence the optimal revenue, by the same
factor. Therefore,
\begin{equation}
\mathrm{OPT}_{v/\bar r}
=
\frac{1}{\bar r}\,\mathrm{OPT}_{r=1}.
\label{eq:scaled-revenue}
\end{equation}
Combining \eqref{eq:det-c-revenue}--\eqref{eq:scaled-revenue} gives
\[
\mathrm{OPT}_{\mathrm{det}}
\ge
\frac{1}{\bar r}\,\mathrm{OPT}_{r=1}.
\]
\qed

\section{Proofs for Section~\ref{sec:pricing-functions} (Pricing Functions)}
\label{app:pricing-functions}

\subsection*{Proof of Theorem~\ref{thm:pricing-reduction}}

The proof has four steps. We first recover a partial price menu from the
mechanism. We then take its lower convex envelope, prove that each type's
selected payment weakly increases, and extend the resulting pricing
function to all allocation probabilities in $[0,1]$.

\paragraph{Step 1: Recover partial menu.}
Let $S\subseteq T$, and let $(x,p)$ be a truthful
mechanism on $S$. Let
\[
\mathcal{X}:=\{x(t):t\in S\}\subseteq[0,1]
\]
be the set of allocation quantities generated by the truthful
mechanism. We first associate a partial pricing function $\pi_0(\cdot)$
with the mechanism.

For every $y\in\mathcal{X}$, choose any type $t$ satisfying $x(t)=y$ and
define
\[
\pi_0(y):=p(t).
\]
$\pi_0(\cdot)$ is well defined. Suppose
$x(t)=x(t')=y$. DSIC implies that $p(t)=p(t')$.

If $0\notin\mathcal{X}$, we extend the domain of $\pi_0$ to
$\mathcal{X}\cup\{0\}$, continue to denote it by $\mathcal{X}$, and set
$\pi_0(0)=0$. This extension does not affect incentive compatibility,
since IR ensures that no type prefers $y=0$ at price $0$ to its assigned
allocation and payment.

The proof of Lemma~\ref{lem:unit-payment-monotonicity} uses only DSIC
and IR comparisons between the two types involved, so the lemma
applies on $S$ as well. Hence $\pi_0$ is nondecreasing on $\mathcal{X}$,
and its unit price $\pi_0(y)/y$ is nondecreasing over positive
$y\in\mathcal{X}$.

For $0\le y_1<y_2$ in $\mathcal{X}$, we also have
\begin{equation}
0\le
\frac{\pi_0(y_2)-\pi_0(y_1)}
     {y_2-y_1}
\le \bar v.
\label{eq:partial-slope-bound}
\end{equation}
To see the upper bound, choose $t_2=(v_2,c_2)\in S$ such that
$x(t_2)=y_2$.
If $y_1=0$, IR gives
\[
\frac{\pi_0(y_2)}{y_2}\le c_2\le v_2\le\bar v.
\]
Suppose $y_1>0$, and choose $t_1\in S$ such that $x(t_1)=y_1$.
Lemma~\ref{lem:unit-payment-monotonicity} and IR imply that type $t_2$
can safely misreport as $t_1$. DSIC therefore gives
\[
v_2y_2-\pi_0(y_2)
\ge
v_2y_1-\pi_0(y_1).
\]
Hence
\[
\frac{\pi_0(y_2)-\pi_0(y_1)}{y_2-y_1}
\le v_2\le\bar v.
\]
The lower bound follows because $\pi_0(y_2)\ge\pi_0(y_1)$ and
$y_2-y_1>0$. Thus \eqref{eq:partial-slope-bound} holds, which means that
$\pi_0$ is $\bar v$-Lipschitz on $\mathcal{X}$.

The set $\mathcal{X}$ need not be closed, so let
\[
D:=\overline{\mathcal{X}}
\]
be its closure. The Lipschitz bound gives a unique
continuous extension of $\pi_0$ from $\mathcal{X}$ to $D$, which we
continue to denote by $\pi_0$. Since $D$ is a closed subset of $[0,1]$,
it is compact and has a maximum, allowing us to construct the lower
convex envelope on $[0,\max D]$ in the next step.
The unit-price monotonicity and the Lipschitz bound remain valid on $D$.
In particular,
\begin{equation}
0\le
\frac{\pi_0(y_2)-\pi_0(y_1)}
     {y_2-y_1}
\le\bar v
\qquad
\text{for all }y_1<y_2,\quad y_1,y_2\in D,
\label{eq:closed-slope-bound}
\end{equation}
and the unit price $\pi_0(y)/y$ is nondecreasing in $y$ for
positive $y\in D$.

\paragraph{Step 2: Construct convex envelope.}
Let
\[
\bar y:=\max D
\]
and let
\[
\widehat{\pi}:[0,\bar y]\to\mathbb{R}
\]
be the lower convex envelope of the partial function $\pi_0$ on $D$;
that is, $\widehat{\pi}$ is the greatest convex function satisfying
\begin{equation}
\widehat{\pi}(y)\le\pi_0(y)
\qquad
\text{for every }y\in D.
\label{eq:convex-envelope}
\end{equation}
Notice that $\widehat{\pi}$ is defined on the entire interval
$[0,\bar y]$, including quantities that are not in $D$ and hence were not
generated by the original mechanism.

Since $\pi_0$ is nonnegative and $\pi_0(0)=0$, the zero function is a
convex minorant of $\pi_0$. Hence $\widehat{\pi}$ is nonnegative and
$\widehat{\pi}(0)=0$. Consequently, $\widehat{\pi}$ is nondecreasing.
Convexity also implies
\begin{equation}
\frac{\widehat{\pi}(x)}{x}
\le
\frac{\widehat{\pi}(y)}{y}
\qquad
(0<x<y\le\bar y).
\label{eq:envelope-average-price}
\end{equation}
Therefore, for a bidder with unit-payment cap $c$, if quantity
$y\in[0,\bar y]$ is affordable, meaning $\widehat{\pi}(y)\le cy$, then
every $x\in[0,y]$ is also affordable because $\widehat{\pi}(x)\le cx$.

\paragraph{Step 3: Compare payments.}
We now show that offering $\widehat{\pi}$ weakly increases the payment of
every type.

Fix a type $t=(v,c)\in S$, and let
\[
y:=x(t),
\qquad
p:=p(t)=\pi_0(y)
\]
denote its allocation and payment under the original mechanism. If $y=0$,
the claim follows immediately from the nonnegativity of $\widehat{\pi}$.
Suppose henceforth that $y>0$. Since
$y\in\mathcal{X}\subseteq D$, there are only two cases.

\noindent
\textbf{Case 1: $\widehat{\pi}(y)=\pi_0(y)$.}

Here the new price at $y$ equals the original payment $p$.
Since the original mechanism is IR,
\[
p\le cy,
\]
so $y$ remains affordable under $\widehat{\pi}$.

We claim that no quantity below $y$ gives the bidder strictly larger
utility under $\widehat{\pi}$.

First consider $a\in D$ with $a<y$ where the new and original prices
agree (called a contact point):
\[
\widehat{\pi}(a)=\pi_0(a).
\]
By the monotonicity of the unit price,
\[
\frac{\pi_0(a)}{a}
\le
\frac{\pi_0(y)}{y}
=
\frac{p}{y}
\le c,
\]
where the case $a=0$ is immediate. Hence the price at $a$ is affordable to type $t$.

If $a\in\mathcal{X}$, applying DSIC directly to the outcome at $a$ gives
\[
va-\widehat{\pi}(a)=va-\pi_0(a)\le vy-p.
\]

If $a\in D\setminus\mathcal{X}$, the quantity $a$ is not offered by the
original mechanism, but it can be approached arbitrarily closely by
offered quantities because $D$ is the closure of $\mathcal{X}$.
Choose such quantities $a_k\in\mathcal{X}$ with $a_k\to a$.
Since $a<y$, we have $a_k<y$ for all sufficiently large $k$.
The unit price $\pi_0(z)/z$ is nondecreasing in $z$, so for each such
$k$ with $a_k>0$,
\[
\frac{\pi_0(a_k)}{a_k}
\le\frac{\pi_0(y)}{y}=\frac{p}{y}\le c.
\]
Thus $\pi_0(a_k)\le ca_k$, meaning that this offered outcome satisfies
type $t$'s cap constraint. A zero quantity is also affordable because
its price is zero. DSIC therefore gives, for all sufficiently large $k$,
\[
vy-p\ge va_k-\pi_0(a_k).
\]
As $k$ increases, $a_k\to a$ and continuity gives
$\pi_0(a_k)\to\pi_0(a)$. Hence the right-hand side converges to
$va-\pi_0(a)$, while the left-hand side stays fixed.
In either case, using $\pi_0(a)=\widehat{\pi}(a)$, we obtain
\[
va-\widehat{\pi}(a)\le vy-p.
\]
Thus choosing any contact point $a<y$ gives no higher utility than
choosing $y$.

It remains to consider quantities below $y$ where the prices do not
agree, or where no original price is defined. Each such quantity $z$
lies between two contact points $a<b\le y$, and the lower convex
envelope is a straight line between them. Thus, writing
$\lambda=(z-a)/(b-a)$, we have
\[
vz-\widehat{\pi}(z)=(1-\lambda)\bigl[va-\widehat{\pi}(a)\bigr]+\lambda\bigl[vb-\widehat{\pi}(b)\bigr]\le vy-p.
\]
The inequality follows because utility at each contact point is
at most $vy-p$, as shown above (with equality if $b=y$).
Consequently, no quantity below $y$ gives strictly larger utility.
Since $y$ is
affordable, the largest affordable utility maximizer $\widehat y$
satisfies $\widehat y\ge y$. Monotonicity of $\widehat{\pi}$ gives
\[
\widehat{\pi}(\widehat y)\ge\widehat{\pi}(y)=p.
\]
Thus the bidder's payment weakly increases.

\noindent
\textbf{Case 2: $\widehat{\pi}(y)<\pi_0(y)$.}

By the standard structure of a one-dimensional lower convex envelope, there
exist contact points
\[
a<y<b,
\qquad a,b\in D,
\]
such that
\[
\widehat{\pi}(a)=\pi_0(a)=:p_a,
\qquad
\widehat{\pi}(b)=\pi_0(b)=:p_b,
\]
and $\widehat{\pi}$ is affine on $[a,b]$. Write
\[
\widehat{\pi}(z)=p_a+s(z-a),
\qquad
z\in[a,b],
\]
where
\[
s:=\frac{p_b-p_a}{b-a}.
\]

As in Case 1, average-price monotonicity makes $a$ affordable, and DSIC
(with the same limiting argument if $a\notin\mathcal X$) gives
\[
v(y-a)\ge p-p_a>s(y-a),
\]
where the strict inequality follows from $p>\widehat\pi(y)$.
Thus $v>s$. By convexity, all slopes to the left of $b$ are at most $s$,
so utility is strictly increasing on $[0,b]$, including at newly
available quantities.

If $p_b\le cb$, then $b$ is affordable and the selected quantity is at
least $b$. Monotonicity of $\widehat\pi$ and $\pi_0$ gives
\[
p_{\widehat\pi}(v,c)\ge\widehat\pi(b)=p_b=\pi_0(b)\ge\pi_0(y)=p.
\]
Otherwise, $b$ is unaffordable while $\widehat\pi(y)<p\le cy$.
Continuity and average-price monotonicity therefore give a largest
affordable quantity $y^c\in(y,b)$ with $\widehat\pi(y^c)=cy^c$.
Since utility is strictly increasing up to $b$, the bidder selects
$y^c$ and pays
\[
p_{\widehat\pi}(v,c)=cy^c\ge cy\ge p.
\]
Thus payment again weakly increases.

We have therefore shown in both cases that
\begin{equation}
p_{\widehat{\pi}}(v,c)\ge p(v,c)
\qquad
\text{for every }(v,c)\in S.
\label{eq:pointwise-dominance}
\end{equation}

\paragraph{Step 4: Extend to full interval.}
It remains only to extend $\widehat{\pi}$ from $[0,\bar y]$ to $[0,1]$ if
$\bar y<1$. By \eqref{eq:closed-slope-bound}, every slope of the lower convex
envelope lies in $[0,\bar v]$. Define
\[
\pi(z)
=
\begin{cases}
\widehat{\pi}(z),
& 0\le z\le\bar y,\\[1mm]
\widehat{\pi}(\bar y)+\bar v(z-\bar y),
& \bar y<z\le1.
\end{cases}
\]
The resulting function is continuous, convex, and nondecreasing.

Since $\pi$ agrees with $\widehat{\pi}$ on $[0,\bar y]$, all quantities
previously available under $\widehat{\pi}$ remain available at the
same prices. If a bidder instead chooses a newly added quantity
$z>\bar y$, then by monotonicity
\[
\pi(z)\ge\pi(\bar y)
\]
and hence its payment is weakly larger than the payment of any quantity in
$[0,\bar y]$. Thus the extension cannot reduce the payment of any type.

Finally,
\[
\widehat{\pi}(\bar y)
\le
\pi_0(\bar y)
\le
\bar v\,\bar y,
\]
and therefore
\[
\pi(1)
=
\widehat{\pi}(\bar y)+\bar v(1-\bar y)
\le
\bar v.
\]
Thus $\pi$ satisfies all the claimed properties and, by
\eqref{eq:pointwise-dominance},
\[
p_{\pi}(v,c)\ge p(v,c)
\qquad
\text{for every }(v,c)\in S.
\]
\qed

\subsection*{Proof of Lemma~\ref{lem:public-valuation-normalization}}

Convexity and \(\pi(0)=0\) imply that the average price
\(\pi(x)/x\) is nondecreasing on \((0,1]\). Indeed, for
\(0<a<b\le1\),
\[
\pi(a)
\le
\frac{a}{b}\pi(b).
\]

Let \(y\) be the largest maximizer of \(vx-\pi(x)\) on \([0,1]\).
No type chooses a quantity larger than \(y\). To see this, suppose
\(x>y\) is affordable to a type with cap \(c\). If \(y>0\), average-price
monotonicity implies
\[
\frac{\pi(y)}{y}
\le
\frac{\pi(x)}{x}
\le c,
\]
so \(y\) is also affordable. Moreover,
since \(y\) is the largest maximizer, no \(x>y\) can also attain the
maximum; hence \(vy-\pi(y)>vx-\pi(x)\). If
\(y=0\), every positive quantity has strictly lower utility than zero.
Thus the claim holds in either case.

When \(y=0\), every type pays zero. Taking \(\widetilde{\pi}\equiv0\)
preserves all payments, has the stated regularity properties, and makes
full quantity a utility maximizer. Suppose henceforth that \(y>0\), and define
\[
\widetilde{\pi}(x)
:=
\frac{\pi(yx)}{y},
\qquad x\in[0,1].
\]
This rescaling preserves continuity, convexity, monotonicity,
nonnegativity, and the condition \(\widetilde{\pi}(0)=0\). Also,
average-price monotonicity gives
\[
\widetilde{\pi}(1)
=
\frac{\pi(y)}{y}
\le
\pi(1)
\le\bar v.
\]

Fix a cap \(c\). Affordability and utility under the two pricing functions satisfy
\[
\widetilde{\pi}(x)\le cx
\quad\Longleftrightarrow\quad
\pi(yx)\le c(yx)
\]
and
\[
vx-\widetilde{\pi}(x)
=
\frac{v(yx)-\pi(yx)}{y}.
\]
Hence quantities under \(\widetilde{\pi}\) correspond exactly, through
\(x\mapsto yx\), to quantities in \([0,y]\) under \(\pi\). Since no type
chooses above \(y\), the largest affordable utility maximizers satisfy
\[
x_{\widetilde{\pi}}(c)
=
\frac{x_\pi(c)}{y}.
\]
The corresponding payment is therefore
\[
p_{\widetilde{\pi}}(v,c)
=
\frac{p_\pi(v,c)}{y}
\ge
p_\pi(v,c),
\]
because \(0<y\le1\).

Finally, since \(y\) maximizes \(vx-\pi(x)\), for every \(x\in[0,1]\),
\[
vy-\pi(y)
\ge
v(yx)-\pi(yx).
\]
Dividing by \(y\) shows that
\[
v-\widetilde{\pi}(1)
\ge
vx-\widetilde{\pi}(x),
\]
so full quantity maximizes \(vx-\widetilde{\pi}(x)\).
\qed
\subsection*{Proof of Theorem~\ref{thm:public-valuation-optimal-pricing}}

By Lemma~\ref{lem:public-valuation-normalization}, it suffices to consider
pricing functions \(\pi\) for which full quantity maximizes
\(vx-\pi(x)\); fix such a function and define \(z=\pi(1)\).
Comparison with zero quantity gives $v-z\ge0$, so $z\le v$.

If \(z<\underline c\), every type can afford full quantity and, by the
largest-maximizer tie-breaking rule, pays \(z\), so revenue equals \(z\).
Thus every \(z<\underline c\) is strictly dominated by
\(z=\underline c\), which yields revenue \(\underline c\). Hence no
optimal \(z\) lies below \(\underline c\), and it suffices to consider
\(z\in[\underline c,v]\).
Because full quantity maximizes utility, its utility \(v-z\) satisfies
\[
vx-\pi(x)\le v-z
\]
for every \(x\in[0,1]\). Equivalently,
\begin{equation}
\pi(x)\ge vx-(v-z).
\label{eq:pv-utility-lower-bound}
\end{equation}

Fix a unit-payment cap \(c<z\), and let \(x(c)\) and \(p(c)\) denote the
chosen quantity and price. Affordability implies
\begin{equation}
p(c)\le cx(c).
\label{eq:pv-affordability}
\end{equation}
Applying~\eqref{eq:pv-utility-lower-bound} at \(x(c)\) gives
\begin{equation}
p(c)=\pi(x(c))
\ge
vx(c)-(v-z).
\label{eq:pv-selected-utility-bound}
\end{equation}
Combining~\eqref{eq:pv-affordability}
and~\eqref{eq:pv-selected-utility-bound}, we obtain
\[
vx(c)-(v-z)\le cx(c),
\]
and therefore
\[
x(c)\le\frac{v-z}{v-c}.
\]
Using affordability once more,
\begin{equation}
p(c)
\le
cx(c)
\le
\frac{c(v-z)}{v-c}.
\label{eq:pv-type-payment-bound}
\end{equation}

Now consider a unit-payment cap \(c\ge z\). Full quantity is
affordable, since
\[
\pi(1)=z\le c.
\]
It is also a global utility-maximizing quantity. By the
tie-breaking rule, the bidder therefore chooses full
quantity and pays \(z\). Hence,
\begin{equation}
\operatorname{Rev}(\pi)
\le
\int_{\underline c}^z
\frac{c(v-z)}{v-c}h(c)\,dc
+
z\bigl(1-H(z)\bigr).
\label{eq:pv-revenue-upper-bound}
\end{equation}
The integral bounds the expected revenue from types with
\(c\in[\underline c,z]\), using the payment bound
in~\eqref{eq:pv-type-payment-bound}. The second term is the expected
revenue from types with \(c\in[z,v]\): each such type pays \(z\), and
these types have probability \(1-H(z)\).

We next show that the upper bound in
\eqref{eq:pv-revenue-upper-bound} is attained by
\[
\pi_z(x)=\bigl[vx-(v-z)\bigr]_+.
\]
To determine the quantity chosen under \(\pi_z\), let
\(x_0=(v-z)/v\). For \(x\le x_0\), we have \(\pi_z(x)=0\), so
utility is \(vx\le vx_0=v-z\); for \(x\ge x_0\), we have
\(\pi_z(x)=vx-(v-z)\), so utility is \(v-z\). Thus every quantity in
\([x_0,1]\) maximizes utility, and the tie-breaking rule selects the
largest affordable quantity in this interval. For \(c<z\), affordability
on this interval is equivalent to
\(vx-(v-z)\le cx\), or \(x\le(v-z)/(v-c)\). Because \(c<z\), this
upper bound is below \(1\), so the largest affordable maximizer is
\((v-z)/(v-c)\). For \(c\ge z\), full quantity is affordable and is the
largest maximizer. The resulting payments are therefore
\[
p_z(c)=
\begin{cases}
\dfrac{c(v-z)}{v-c}, & c<z,\\[4pt]
z, & c\ge z.
\end{cases}
\]
Thus $\operatorname{Rev}(\pi)\le\operatorname{Rev}(\pi_z)$, and it remains
to optimize over $z$. For $z\in[\underline c,v)$, using
\(H(z)=\int_{\underline c}^z h(c)\,dc\), the expected revenue is
\[
\begin{aligned}
R(z):=\operatorname{Rev}(\pi_z)
&=\int_{\underline c}^z\frac{c(v-z)}{v-c}h(c)\,dc
  +z\bigl(1-H(z)\bigr)\\
&=z-v\int_{\underline c}^z\frac{z-c}{v-c}h(c)\,dc.
\end{aligned}
\]
Differentiating with respect to \(z\) using the Leibniz rule gives
\[
R'(z)
=1-v\int_{\underline c}^z\frac{h(c)}{v-c}\,dc.
\]
As \(h>0\), \(R'\) is continuous and strictly
decreasing, with \(R'(\underline c)=1\). Since
\[
\int_{\underline c}^{v}h(c)\,dc
=1
>
\frac{v-\underline c}{v},
\]
there exists some \(z_0<v\), sufficiently close to \(v\), such that
\[
\int_{\underline c}^{z_0}h(c)\,dc
>
\frac{v-\underline c}{v}.
\]
Therefore,
\[
R'(z_0)
=1-v\int_{\underline c}^{z_0}\frac{h(c)}{v-c}\,dc
\le 1-\frac{v}{v-\underline c}
\int_{\underline c}^{z_0}h(c)\,dc<0.
\]
Hence there is a unique $c^*\in(\underline c,v)$ satisfying
$R'(c^*)=0$. Since $R'$ is strictly decreasing, $c^*$ uniquely
maximizes $R$ on $[\underline c,v)$.
When $z=v$, only type $c=v$ can make a positive
payment. Since the cap distribution is atomless, this type has probability
zero, so $R(v)=0$. As shown above,
$z=\underline c$ yields positive revenue, so $z=v$ cannot be optimal and
$c^*$ uniquely maximizes $R$ on $[\underline c,v]$.
Finally, to solve for \(c^*\), we rewrite the first-order condition
\(R'(c^*)=0\) in terms of the cap distribution \(H\) and its density \(h\):
\[
\begin{aligned}
0=R'(c^*)
&=1-v\int_{\underline c}^{c^*}\frac{h(c)}{v-c}\,dc\\
&=1-\int_{\underline c}^{c^*}
\left(1+\frac{c}{v-c}\right)h(c)\,dc\\
&=1-H(c^*)-
\int_{\underline c}^{c^*}\frac{c}{v-c}h(c)\,dc.
\end{aligned}
\]
Hence $c^*$ is the unique solution of
\[
\int_{\underline c}^{c^*}\frac{c}{v-c}h(c)\,dc
=1-H(c^*).
\]
\qed
\subsection*{Proof of Theorem~\ref{thm:public-roi-optimal-pricing}}

Let \(S(v)=1-F(v)\) for \(v\in[0,\bar v]\).
When \(r=1\), the result follows immediately. In this case, the power-law pricing function reduces to the linear function
\[
\pi^*(x)=v^*x,
\]
which implements the standard monopoly posted-price mechanism. The optimal \(v^*\) satisfies
\[
S(v^*)-v^*f(v^*)=0,
\]
which is the standard Myerson reserve price.
In the remainder of the proof, we assume \(r>1\).

We first express the revenue of a power-law pricing function in terms of its scale
and derive bounds under the DMR condition. We then normalize an arbitrary convex
pricing function and use those bounds to compare its revenue with a power-law
pricing function. Finally, we optimize the scale.

\paragraph{Step 1: Power-law bound.}
For each \(z\in[0,\bar v]\), consider the power-law pricing function
\[
\pi_{z}(x):=\frac{z}{r}x^{r}.
\]
Its marginal-price and affordability cutoffs coincide for every $x>0$:
\[
\pi_{z}'(x)
=
zx^{r-1}
=
r\frac{\pi_{z}(x)}{x}.
\]
Let
\[
B(z):=\operatorname{Rev}(\pi_{z})
=\int_{0}^{1}\pi_{z}'(x)S\!\left(\pi_{z}'(x)\right)\,dx
=\int_{0}^{1}zx^{r-1}S\!\left(zx^{r-1}\right)\,dx.
\]
Since \(f\) is continuous, Leibniz's rule gives
\begin{equation}
B'(z)
=
\int_{0}^{1}
x^{r-1}\left[
S\!\left(zx^{r-1}\right)
-zx^{r-1}f\!\left(zx^{r-1}\right)
\right]\,dx.
\label{eq:B-derivative}
\end{equation}

We first derive an equation relating \(B\) and \(B'\). Differentiating
\(zx^{r}S(zx^{r-1})\) with respect to \(x\) yields
\[
\begin{aligned}
\frac{d}{dx}
\left[
zx^{r}S\!\left(zx^{r-1}\right)
\right]
&=
zx^{r-1}S\!\left(zx^{r-1}\right)
\\
&\quad+
(r-1)z x^{r-1}
\left[
S\!\left(zx^{r-1}\right)
-zx^{r-1}f\!\left(zx^{r-1}\right)
\right].
\end{aligned}
\]
Integrating over \(x\in[0,1]\) therefore gives
\begin{equation}
B(z)+(r-1)zB'(z)=zS(z).
\label{eq:B-identity}
\end{equation}

We next establish the bounds
\begin{equation}
S(z)-zf(z)\leq rB'(z)\leq S(z)
\qquad
\text{for every }z\in[0,\bar v].
\label{eq:B-prime-bounds}
\end{equation}
By DMR, \(S(v)-vf(v)\) is strictly decreasing on \([0,\bar v]\). Since
\(zx^{r-1}\leq z\), \eqref{eq:B-derivative} implies
\[
B'(z)
\geq
\bigl[S(z)-zf(z)\bigr]\int_{0}^{1}x^{r-1}\,dx
=
\frac{S(z)-zf(z)}{r}.
\]
This proves the first inequality in \eqref{eq:B-prime-bounds}.

For the second inequality, since \(S\) is nonincreasing, we have
\[
B(z)
=
\int_{0}^{1}
zx^{r-1}S\!\left(zx^{r-1}\right)\,dx
\geq
S(z)\int_{0}^{1}zx^{r-1}\,dx
=
\frac{zS(z)}{r}.
\]
Combining this inequality with \eqref{eq:B-identity}, for \(z>0\),
\[
(r-1)zB'(z)
=
zS(z)-B(z)
\leq
\frac{r-1}{r}zS(z),
\]
and hence \(rB'(z)\leq S(z)\).
Since \(B'\) and \(S\) are continuous, letting \(z\downarrow0\) shows
that the inequality also holds at \(z=0\).

\paragraph{Step 2: Normalize an arbitrary pricing function.}
Fix an arbitrary pricing function \(\pi\). We normalize \(\pi\) so that
its marginal-price and affordability cutoffs lie in \([0,\bar v]\), so
that \eqref{eq:B-prime-bounds} applies to them. Let \(y\) be the quantity selected by the
highest-valuation type \(\bar v\). Suppose a type \(v<\bar v\) selected
some \(x>y\). Since \(x\) is affordable to \(v\), it is also
affordable to \(\bar v\); moreover, \(y\) is affordable to
\(v\), trivially if \(y=0\) and otherwise by average-price monotonicity.
Since \(\bar v\) selects \(y\),
\[
\begin{aligned}
\bigl[vy-\pi(y)\bigr]-\bigl[vx-\pi(x)\bigr]
&=
\bigl[\bar v y-\pi(y)\bigr]
-\bigl[\bar v x-\pi(x)\bigr]\\
&\quad+(\bar v-v)(x-y)
>0,
\end{aligned}
\]
contradicting the selection of \(x\) by \(v\). Thus no type $v<\bar v$
selects a quantity above \(y\).

If \(y=0\), no type selects a positive quantity, so
\(\operatorname{Rev}(\pi)=0=\operatorname{Rev}(\pi_0)\), so the theorem
holds in this case. We may therefore assume \(y>0\) and define
\[
\widetilde\pi(s):=\frac{\pi(ys)}{y},
\qquad s\in[0,1].
\]
The normalized pricing function \(\widetilde\pi\) is continuous, convex,
and nondecreasing. For any pricing
function \(\rho\), write \(x_\rho(v)\) and \(p_\rho(v)\) for the quantity
and payment selected by type \(v\). For every
valuation \(v\) and \(s\in[0,1]\),
\[
\widetilde\pi(s)\leq\frac{v}{r}s
\quad\Longleftrightarrow\quad
\pi(ys)\leq\frac{v}{r}(ys),
\]
and
\[
vs-\widetilde\pi(s)
=
\frac{v(ys)-\pi(ys)}{y}.
\]
Thus the map \(s\mapsto ys\) is a bijection between quantities affordable
under \(\widetilde\pi\) and quantities in \([0,y]\) affordable under
\(\pi\), and it preserves their utility ranking. Therefore
\[
x_{\widetilde\pi}(v)=\frac{x_\pi(v)}{y}.
\]
Consequently,
\[
p_{\widetilde\pi}(v)
=\widetilde\pi\bigl(x_{\widetilde\pi}(v)\bigr)
=\frac{\pi\bigl(yx_{\widetilde\pi}(v)\bigr)}{y}
=\frac{p_\pi(v)}{y}\geq p_\pi(v),
\]
where the inequality uses \(0<y\leq1\) and \(p_\pi(v)\geq0\).
Thus \(\widetilde\pi\) earns weakly more revenue. Moreover,
average-price monotonicity makes every quantity below \(y\) affordable
to type \(\bar v\). Because type \(\bar v\) selects \(y\), its utility at
\(y\) is at least its utility at \(y-\delta\):
\[
\bar v y-\pi(y)
\geq
\bar v(y-\delta)-\pi(y-\delta).
\]
Rearranging gives
\[
\frac{\pi(y)-\pi(y-\delta)}{\delta}\le\bar v
\]
for every \(\delta\in(0,y]\). By convexity, this implies
\(\pi'(x)\leq\bar v\) for almost every \(x\in(0,y)\). Moreover,
\[
\widetilde\pi'(s)
=\frac{d}{ds}\left[\frac{\pi(ys)}{y}\right]
=\frac{1}{y}\pi'(ys)\frac{d(ys)}{ds}
=\pi'(ys)
\]
for almost every \(s\in(0,1)\). Affordability of \(y\) also gives
\(r\pi(y)/y\leq\bar v\), so we obtain
\[
r\widetilde\pi(1)\le\bar v
\qquad\text{and}\qquad
\widetilde\pi'(s)\le\bar v
\quad\text{for almost every }s\in(0,1).
\]
It therefore suffices to consider the normalized pricing functions
constructed above, which satisfy these two bounds. Suppress the tilde and write
\[
u(x)=\pi'(x),
\qquad
w(x)=r\frac{\pi(x)}{x}.
\]
Then \(u(x),w(x)\in[0,\bar v]\) for almost every \(x\in(0,1)\).

\paragraph{Step 3: Compare revenues.}
At each differentiability point \(x\), a type selects at least \(x\) if
and only if \(v\geq\max\{u(x),w(x)\}\), which occurs with probability
\(S(\max\{u(x),w(x)\})\). Hence
\[
\operatorname{Rev}(\pi)
=\int_0^1
u(x)S\!\left(\max\{u(x),w(x)\}\right)\,dx.
\]

For any point \(x\) at which \(\pi\) is differentiable, abbreviate
\(u(x)\) and \(w(x)\) to \(u\) and \(w\), respectively. Because
\(w\in[0,\bar v]\), \eqref{eq:B-prime-bounds}, which holds for every
scalar argument \(z\in[0,\bar v]\), applies with \(z=w\). If \(u\leq w\), the upper bound in
\eqref{eq:B-prime-bounds} gives \(rB'(w)\leq S(w)\). Since
\(u-w\leq0\), multiplying by \(u-w\) reverses the inequality:
\[
rB'(w)(u-w)\geq S(w)(u-w).
\]
Adding \(wS(w)\) to both sides gives
\[
wS(w)+rB'(w)(u-w)
\geq wS(w)+S(w)(u-w)=uS(w).
\]
If \(u\geq w\), DMR implies that, for every \(t\in[w,u]\),
\[
S(t)-tf(t)\leq S(w)-wf(w).
\]
Since \(\frac{d}{dt}[tS(t)]=S(t)-tf(t)\), the fundamental theorem of
calculus gives
\[
\begin{aligned}
uS(u)-wS(w)
&=\int_w^u \frac{d}{dt}\bigl[tS(t)\bigr]\,dt\\
&=\int_w^u \bigl[S(t)-tf(t)\bigr]\,dt\\
&\leq\int_w^u \bigl[S(w)-wf(w)\bigr]\,dt\\
&=\bigl[S(w)-wf(w)\bigr](u-w).
\end{aligned}
\]
The lower bound in \eqref{eq:B-prime-bounds} and \(u-w\geq0\) give
\[
\bigl[S(w)-wf(w)\bigr](u-w)
\leq rB'(w)(u-w).
\]
Combining these inequalities and adding \(wS(w)\) gives
\[
uS(u)
\leq
wS(w)+rB'(w)(u-w).
\]
Thus, in either case,
\[
uS\!\left(\max\{u,w\}\right)
\leq wS(w)+rB'(w)(u-w)
\]
for almost every \(x\in(0,1)\).

Integrating the pointwise inequality above gives
\[
\begin{aligned}
\operatorname{Rev}(\pi)
&\leq \int_0^1\left[wS(w)+rB'(w)(u-w)\right]\,dx\\
&=\int_0^1\left[B(w)+(ru-w)B'(w)\right]\,dx\\
&=\int_0^1\left[B(w)+xB'(w)w'\right]\,dx\\
&=\int_0^1\frac{d}{dx}\left[xB(w(x))\right]\,dx\\
&=\left[xB(w(x))\right]_{0^+}^{1}
=B(w(1)).
\end{aligned}
\]
The inequality uses the revenue expression above, and the first equality
uses \eqref{eq:B-identity}. For the second equality, both \(\pi\) and
\(w\) are locally absolutely continuous on \((0,1]\), and for almost
every \(x\in(0,1)\),
\begin{equation}
w'(x)
=
r\frac{x\pi'(x)-\pi(x)}{x^{2}}
=
\frac{ru(x)-w(x)}{x}.
\label{eq:cutoff-derivative}
\end{equation}
The third equality follows from the chain rule. The integrals are
interpreted as improper at \(0\), and the endpoint evaluation follows
because \(B\) is bounded on \([0,\bar v]\), so
\(\lim_{x\downarrow0}xB(w(x))=0\). Therefore,
\begin{equation}
\operatorname{Rev}(\pi)
\leq B(w(1)).
\label{eq:global-upper-bound}
\end{equation}

Set
\[
z:=w(1)=r\pi(1).
\]
The power-law pricing function
\[
\widehat{\pi}(x):=\frac{z}{r}x^{r}
\]
has the same terminal price as \(\pi\), since
\[
\widehat{\pi}(1)=\frac{z}{r}=\pi(1),
\]
and it is feasible because \(z\leq\bar v\). By definition,
\[
\operatorname{Rev}(\widehat{\pi})=B(z)=B(w(1)).
\]
Hence \eqref{eq:global-upper-bound} implies
\begin{equation}
\operatorname{Rev}(\pi)
\leq
\operatorname{Rev}(\widehat{\pi}).
\label{eq:power-dominates}
\end{equation}
Thus the power-law pricing function earns weakly more revenue than any
normalized pricing function with the same terminal price. Together with
the normalization above, this shows that some power-law pricing
function earns weakly more revenue than any pricing
function.

\paragraph{Step 4: Optimize scale.}
It remains to optimize over \(z\). By \eqref{eq:B-derivative},
\[
B'(z)
=
\int_{0}^{1}
x^{r-1}\left[
S\!\left(zx^{r-1}\right)
-zx^{r-1}f\!\left(zx^{r-1}\right)
\right]\,dx.
\]
By DMR, \(S(v)-vf(v)\) is strictly decreasing, so \(B'\) is strictly
decreasing. Moreover, \(B'(0)=1/r>0\).
At \(z=\bar v\), Equation~\eqref{eq:B-identity} and \(S(\bar v)=0\)
give \(B(\bar v)+(r-1)\bar v B'(\bar v)=0\).
Since \(B(\bar v)>0\), it follows that \(B'(\bar v)<0\). Thus there is
a unique \(v^{*}\in(0,\bar v)\) satisfying
\[
B'(v^{*})
=
\int_{0}^{1}
x^{r-1}\left[
S\!\left(v^{*}x^{r-1}\right)
-v^{*}x^{r-1}f\!\left(v^{*}x^{r-1}\right)
\right]\,dx
=0.
\]
\qed

\clearpage
\bibliographystyle{plain}
\bibliography{Auction}

\end{document}